\documentclass[a4paper,UKenglish,cleveref, autoref, thm-restate]{lipics-v2021}

\nolinenumbers

\title{4-Swap: Achieving Grief-Free and Bribery-Safe Atomic Swaps Using Four Transactions} 

\titlerunning{4-Swap: Grief-Free and Bribery-Safe Atomic Swaps Using Four Transactions}

\author{Kirti Singh}{Indian Institute of Technology Bombay, India \and Institute for Development and Research in Banking Technology, Hyderabad, India }{kirtisingh@cse.iitb.ac.in}{https://orcid.org/0009-0006-1846-8682}{}

\author{Vinay J. Ribeiro}{Indian Institute of Technology Bombay, India }{vinayr@iitb.ac.in}{https://orcid.org/0000-0001-5627-5343}{}

\author{Susmita Mandal}{Institute for Development and Research in Banking Technology, Hyderabad, India }{msusmita@idrbt.ac.in}{https://orcid.org/0009-0006-1846-8682}{}

\authorrunning{K. Singh, V.\,J. Ribeiro and S. Mandal}

\Copyright{Kirti Singh, Vinay J. Ribeiro and Susmita Mandal}



\category{}

\relatedversion{}

\begin{document}

\maketitle

\begin{abstract}
Cross-chain asset exchange is crucial for blockchain interoperability. Existing solutions rely on trusted third parties and risk asset loss, or use decentralized alternatives like atomic swaps, which suffer from grief attacks. Griefing occurs when a party prematurely exits, locking the counterparty's assets until a timelock expires. Hedged Atomic Swaps mitigate griefing by introducing a penalty premium; however, they increase the number of transactions from four (as in Tier Nolan's swap) to six, which in turn introduces new griefing risks. Grief-Free (GF) Swap reduces this to five transactions by consolidating assets and premiums on a single chain. However, no existing protocol achieves grief-free asset exchange in just four transactions.
    
This paper presents 4-Swap, the first cross-chain atomic swap protocol that is both grief-free and bribery-safe, while completing asset exchange in just four transactions. By combining the griefing premium and principal into a single transaction per chain, 4-Swap reduces on-chain transactions, leading to faster execution compared to previous grief-free solutions. It is fully compatible with Bitcoin and operates without the need for any new opcodes. A game-theoretic analysis shows that rational participants have no incentive to deviate from the protocol, ensuring robust compliance and security.
\end{abstract}

\section{Introduction}
\label{sec:introduction}

The trade of cross-chain assets has been a significant problem in blockchains. Centralized exchanges such as Binance, Coinbase, and CoinDCX solve this problem by taking ownership of users' assets and exchanging them among the interested parties. Using centralized exchanges makes the exchange efficient and easy. However, the involvement of a trusted third party leads to the risk of users losing their assets, as happened in the MTGox hack~\cite{DBLP:conf/service/DingC22, rao2022mt} and FTX collapse~\cite{DBLP:conf/fc/FuWYC23}. Thus, there is a need for trustless exchange protocols. Some cross-chain decentralized exchange solutions exist, such as two-way pegged tokens~\cite{DBLP:journals/csur/BelchiorVGC22} and using a third chain~\cite{DBLP:journals/corr/abs-2110-13871}. However, they face some issues: pegged tokens do not give native tokens, and using a third chain increases the overall overhead for the exchange. 
    
    Atomic swaps solve the issues by providing native tokens without needing a third coordinating chain. Tier Nolan introduced it in a Bitcoin forum in 2013~\cite{tnswap, nakamoto2008bitcoin}. It typically involves four transactions, the locking of tokens by the two parties, which later claim each others' tokens before a timelock expires. Atomic swaps primarily use Hashed Time-Locked Contracts (HTLCs), which provide a mechanism to lock assets on the blockchain with a hashlock and to refund the locked assets after the timelock expires. However, they can also leverage multi-signature transactions~\cite{DBLP:conf/podc/Herlihy18, DBLP:conf/esorics/HoenischMMR22}, adaptor signatures~\cite{DBLP:conf/esorics/HoenischMMR22, DBLP:conf/sp/ThyagarajanMM22}, and verifiable timed signatures~\cite{DBLP:conf/sp/ThyagarajanMM22} for implementation. Although Atomic Swap provides a decentralized exchange solution, it suffers from the problem of grief.
    
    Griefing is a common problem in Atomic Swaps wherein a party who locks the tokens waits for the timelock to expire when the counterparty abandons the exchange. To deter this premiums are used as collateral. In 2021, Xue and Herlihy~\cite{DBLP:conf/podc/XueH21} used the concept of cross-locking of griefing premium in the two chains to solve griefing. However, this further led to griefing on the locked griefing premium, and the number of transactions in the swap increased to six. 
    
    In 2022, Nadahalli \textit{\textit{et al.}} proposed combining the locking of griefing premium with the assets in one transaction~\cite{DBLP:conf/icbc2/NadahalliKW22}. Thus, the problem of griefing on the premium is solved. However, they failed to combine the premium on both chains, requiring five transactions for a grief-free atomic swap. Mazumdar also solves the griefing problem by allowing the participating parties to leave the exchange whenever they like, increasing the total number of transactions to eight~\cite{DBLP:conf/tpsisa/Mazumdar22}. This leads to our question: is achieving a grief-free atomic swap in only four transactions possible?

    We present the 4-Swap (or 4S, for short) protocol in two versions designed to simplify the exposition of this complex protocol. By breaking it down into incremental versions, we aim to make its functionality and security properties more accessible.
    We begin with 4S-v1, which operates under the simplifying assumption of zero delay in transactions being added to the blockchain. In this idealized setting, a published transaction is immediately included without waiting in the mempool, effectively eliminating the risks of bribery or censorship attacks. This version allows us to focus on the core mechanics of the protocol, which achieves grief-free execution in just four steps. It does so by leveraging cross-locking of principals and premiums, as seen in prior works~\cite{DBLP:conf/podc/XueH21, DBLP:conf/icbc2/NadahalliKW22}. We introduce two novel concepts of cross-publishing of locking transactions and a new timelock in the initiating chain for early refunds. These concepts help combine the principal and griefing premiums for both parties in both chains, reducing the total number of transactions to just four.
    
    Building on this foundation, we present the next version that removes the zero-delay assumption, thus accounting for real-world conditions where transaction delays and susceptibility to bribery attacks are practical concerns. To address this, we integrate the concept of mutual destruction from~\cite{DBLP:conf/sp/TsabaryYME21} coupled in both the chains, ensuring the protocol remains bribery-safe. The final version retains its simplicity and efficiency, maintaining grief-free and bribery safety within a four-step process. We do not claim the protocol to be bribery free as it does not considers the reverse bribery attacks where the miners incentivizes the involved parties to publish the conflicting transactions.
    The protocol's game-theoretical analysis demonstrates that it consistently follows its intended execution path under the assumption of rational behaviour from all parties. We perform a security analysis of the protocol using the UC Framework~\cite{DBLP:conf/tcc/CanettiDPW07} (See Appendix~\ref{securityanalysis}) and show the possible redeem paths of the protocol. The scripts presented in Appendix~\ref{implementation} show that the protocol can be implemented in Bitcoin without needing additional opcodes.

    \paragraph*{Our Contributions}
    \begin{enumerate}
        \item We propose 4-Swap, the first grief-free and bribery-safe atomic swap protocol that uses only four transactions, matching the number of transactions in the TN-Swap.
        \item We then present a game-theoretical analysis of our protocol to show that participants have no incentive to deviate from the protocol under rational behaviour.
        \item We perform a security analysis of the protocol using the UC framework~\cite{DBLP:conf/tcc/CanettiDPW07} in Appendix~\ref{securityanalysis}.
    \end{enumerate}

\section{Background}

\subsection{Atomic Swaps - TN Swap}
\begin{figure}
    \centering
    \includegraphics[width=0.8\textwidth]{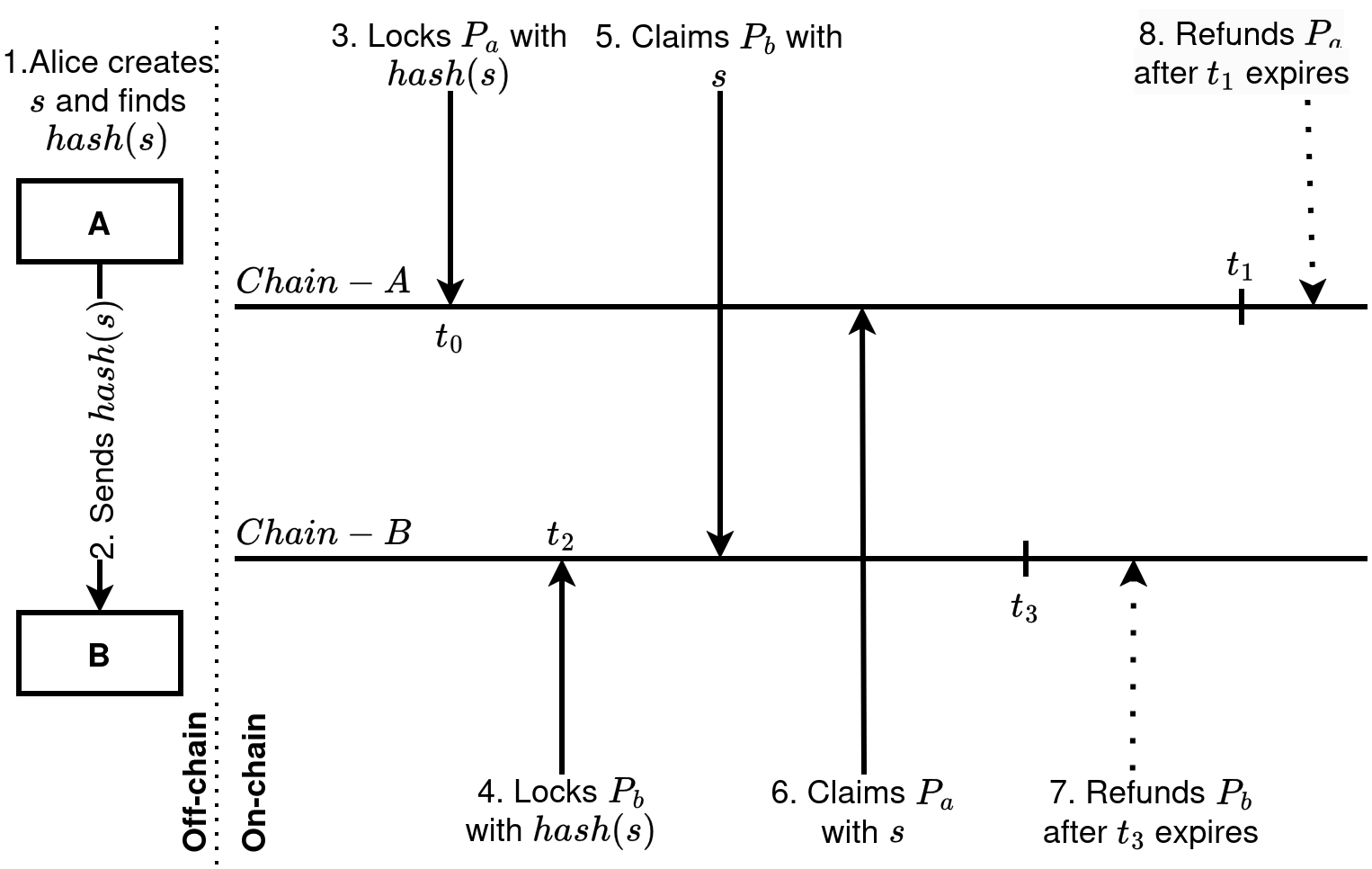}
    \caption{Atomic Swap Protocol}
    \label{fig:atomicswap}
\end{figure}
Atomic swap is a decentralized method of performing cross-chain exchanges. Tier Nolan first introduced it in 2013~\cite{tnswap}. It typically uses two HTLCs with the same hash to lock the tokens. The parties on the successful locking claim each other's tokens before the set timeout expires in their respective chains~\cite{herlihy2018atomic}. 

Consider $A$ and $B$, two parties interested in exchanging tokens \( P_a \) on $Chain\text{-}A$ and \( P_b \) on $Chain\text{-}B$, respectively. They must follow the following steps for a successful atomic swap exchange, as shown in Figure~\ref{fig:atomicswap}.
\begin{enumerate}
	\item $A$ generates secret \( s \), finds $\mathcal{H}(s)$ and sends $\mathcal{H}(s)$ to $B$.
	\item $A$ locks her tokens on $Chain\text{-}A$ with $\mathcal{H}(s)$ setting up a refund timelock of \( t_1 \).
	\item $B$ locks his tokens on $Chain\text{-}B$ with same $\mathcal{H}(s)$ setting a refund timelock of \( t_3 \).
	\item $A$ claims $B$'s locked tokens on $Chain\text{-}B$ with \( s \). 
	\item $B$ uses the revealed preimage \( s \) to claim $A$'s tokens.
	\item If $A$ does not claim $B$'s tokens, $B$ can refund after \( t_3 \).
	\item If $B$ does not claim $A$'s tokens, $A$ can refund after \( t_1 \).
\end{enumerate}

The protocol ensures that no party loses their tokens in the exchange. Moreover, the gap in two timelocks ensures that $B$ gets enough time to claim $A$'s token when $A$ has already claimed $B$'s token. Although atomic swap solves the problem of decentralized cross-chain asset exchange, it gives rise to the problem of griefing.

\subsubsection{Griefing Attack}
\label{grief}
Griefing occurs when a party locks the tokens, and the counterparty exits the exchange, forcing the counterparty to wait until the refund timelock expires to retrieve their tokens. 

In Figure~\ref{fig:atomicswap}, there are two possible scenarios of griefing.
\begin{itemize}
	\item $A$ locks tokens, but $B$ leaves the exchange. $A$ waits until \( t_1 \) to refund the locked tokens.
	\item $A$ and $B$ lock their tokens, but $A$ leaves without claiming $B$'s tokens. Since $A$ only knows the preimage \( s \), $B$ waits until the refund timelock \( t_3 \) expires to get back his tokens.
\end{itemize}

\subsubsection{Bribery Attacks}
\label{bribe}
Malicious parties bribe the miners to censor some transactions in the blockchain to gain an unfair advantage. HTLCs suffer from bribery attacks~\cite{10.1007/978-3-662-64322-8_3}. Suppose $A$ pays $B$ using HTLC and bribes the miners to censor B's claim transaction until the refund timelock expires. After the timelock expires, $A$ can refund the locked tokens, and $B$ fails to get the tokens. Since atomic swaps also use HTLCs, they are prone to bribery attacks. Bribery usually happens by either a party putting up a conflicting transaction(like the refund in HTLC) with a very high transaction fee or putting the refund transaction along with a future transaction that takes input from the refund transaction and can be redeemed by anyone(particularly miners).

In Figure~\ref{fig:atomicswap}, there are two possible scenarios of bribery attacks.
\begin{itemize}
    \item\ $A$ and $B$ lock their tokens, but $B$ bribes the miners to censor A's claim transaction until \(t_3\). Since $B$ can see A's claim transaction, he fetches \(s\) and claims A's tokens \(P_a\). After \(t_3\), he refunds his tokens \(P_b\) and $A$ gets nothing.
    \item $A$ claims $B$ 's tokens after both successfully locked them. She then bribes the miners to censor $B$ 's claim transaction until \(t_1\). After \(t_1\), $A$ refunds her locked tokens \(P_a\), and $B$ gets nothing.
\end{itemize}

\subsubsection{Bribery Attacks Solutions for HTLCs}
In 2021, Tsabary \textit{et al.}~\cite{DBLP:conf/sp/TsabaryYME21} proposed MAD-HTLC to remove bribery in HTLC. They modified the HTLC by adding a premium to be locked by the sender in addition to the principal value. The idea was to allow the miners to slash all the locked tokens, disincentivizing the parties to attempt bribery. 

Suppose there are two parties where $A$ is paying $B$ using HTLC. The solution works as follows:

\begin{itemize}
    \item $A$ locks \(X\) (principal) and \(Y\) (premium) with \(H(pre_a)\). $Y$ can be refunded after timelock.
    \item To refund \(X\), $A$ must reveal \(pre_b\) after the timelock.
    \item  $B$ claims \(X\) before the timelock by revealing \(pre_a\).
    \item If $A$ attempts to censor $B$'s claim by bribing miners with a high-fee refund transaction, miners can extract \(pre_b\) from $A$'s refund transaction. Anyone with \(pre_a, pre_b\) can claim \(X + Y\), so miners take the funds. Bribery fails; $A$ gains nothing.
\end{itemize}

The sender can refund the premium after the timelock expires. Also, to refund the principal tokens locked by the hash of \(pre_a\), the sender has to release a new preimage \(pre_b\) after the timelock expires. By bribing miners, the sender can still try to censor the receiver's claimed transaction, which reveals \(pre_a\). However, if the sender tries to bribe the miners by putting in a refund transaction with very high transaction fees, the miners can fetch the \(pre_b\) from it. Since MAD-HTLC allows anyone with \(pre_a\) and \(pre_b\) to claim the tokens and the premium, miners can create their transaction and try to create a block with it. Thus, the sender does not get any benefit from bribing.

\subsubsection{Reverse Bribery Attacks}
MAD-HTLC does not address scenarios where miners incentivize protocol participants to behave maliciously, a situation known as a reverse bribery attack. This attack was introduced by Wadhwa \textit{et al.}~\cite{DBLP:conf/ndss/WadhwaSZN23} in 2022, along with the proposal of He-HTLC, a protocol designed to mitigate such threats. Unlike MAD-HTLC, He-HTLC introduces a delay in the sender's refund process and allows multiple miners to collectively punish a malicious sender. Additionally, when a miner redeems funds using $pre_a$ and $pre_b$, they receive only a portion of the locked tokens; the remainder is burned. This partial reward mechanism discourages collusion between miners and protocol participants.

However, He-HTLC does not offer a construction for atomic swaps within its framework. Rapidash~\cite{chung2022rapidash}, also published in 2022, builds on similar ideas. It introduces a two-step refund mechanism and extends the protocol to support atomic swaps. While this approach solves the reverse bribery attacks, it introduces increased on-chain transaction overhead and fails to be grief-free as one party can abandon the protocol after the counterparty has already locked their assets.

\section{Related Works}
    We now discuss the related work of minimizing the number of on-chain transactions of atomic swaps while keeping it grief-free. Both works assume that parties have some tokens in the chain they want to receive tokens to safeguard from griefing attacks.
    \subsection{Hedged Atomic Swap}
    \label{hedgedswap}
    Xue and Herlihy~\cite{DBLP:conf/podc/XueH21} in 2021 proposed a Hedged Atomic Swap protocol that tried to solve the problem of griefing attacks using griefing premiums. The parties have to cross-lock the premium amounts with the principal amounts of the counterparty. The premiums serve as an amount to compensate the parties from griefing. Thus, it introduces the assumption that parties need a small amount of tokens in the other chain where they want to exchange them. The cross-locking helps get the refund premium and claim principal in one step. However, since the protocol requires premium locking, the number of transactions rose to six, two more than the TN-Swap.

    \begin{figure}
        \centering
        \includegraphics[width=0.7\textwidth]{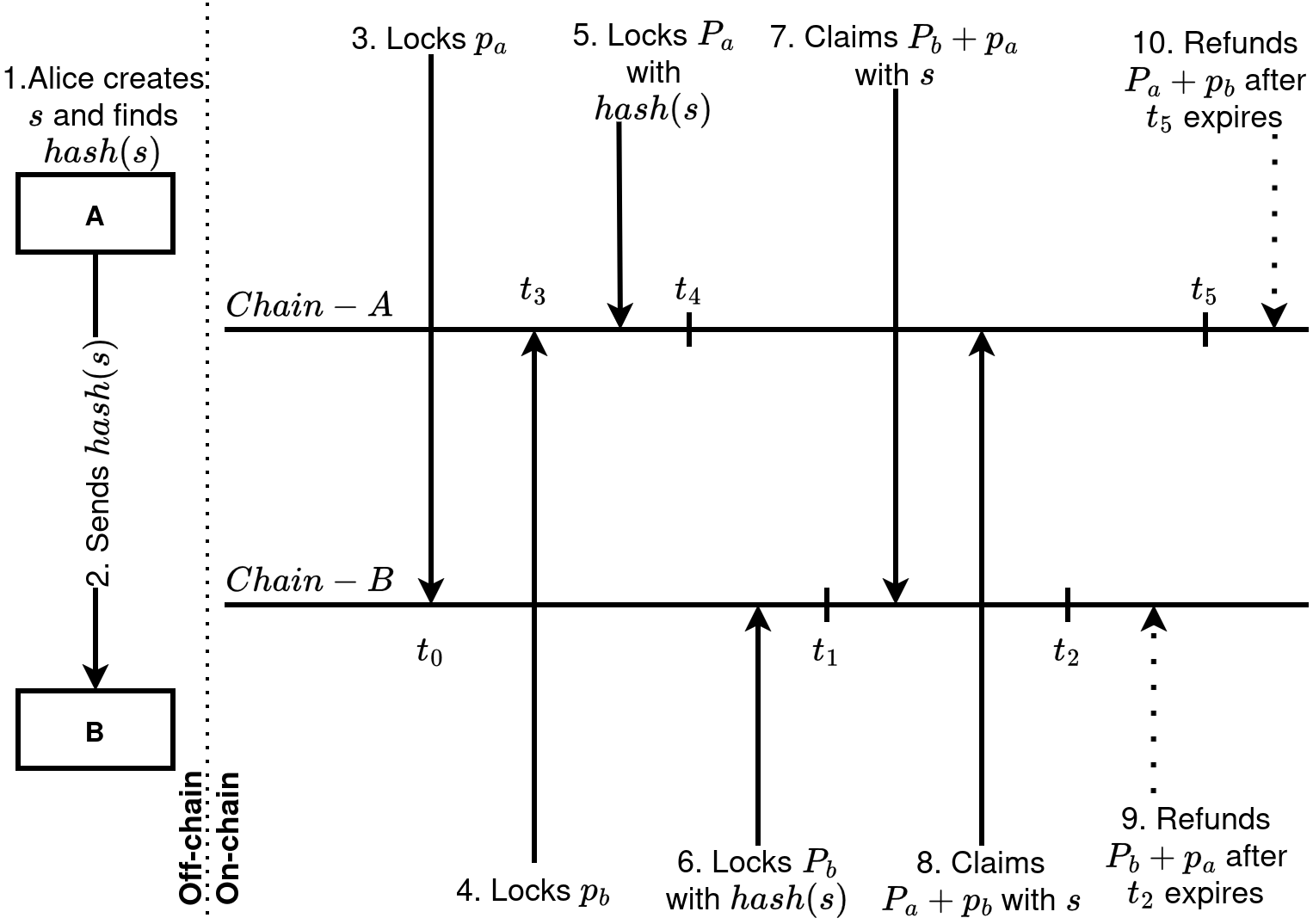}
        \caption{Hedged Atomic Swap Protocol}
        \label{fig:hedgedatomicswap}
    \end{figure}
    
    Consider $A$ and $B$, two parties interested in exchanging tokens \( P_a \) on $Chain\text{-}A$ and \( P_b \) on $Chain\text{-}B$, respectively. $A$ must have $p_a$ on $Chain\text{-}B$ and $B$ must have $p_b$ on $Chain\text{-}A$ as griefing premium. They must follow the following steps for a successful atomic swap exchange, as shown in Figure~\ref{fig:hedgedatomicswap}.
    \begin{enumerate}
        \item $A$ generates a secret \( s \), finds $\mathcal{H}(s)$.
        \item $A$ sends $\mathcal{H}(s)$ to B.
        \item $A$ locks her premium $p_a$ on $Chain\text{-}B$ with $\mathcal{H}(s)$ setting up a timeout of $t_1$ and refund timelock of \( t_2 \) on $Chain\text{-}B$. $B$ should lock his principal before $t_1$ on $Chain\text{-}B$.
        \item $B$ locks his premium $p_b$ on $Chain\text{-}A$ with $\mathcal{H}(s)$ setting up a timeout of \( t_4 \) and a refund timelock of $t_5$ on $Chain\text{-}A$. $t_4$ is the timeout for $A$ to lock her principal.
        \item $A$ locks her tokens on $Chain\text{-}A$ with $\mathcal{H}(s)$ before $t_4$. This locks $P_a+p_b$ until timelock $t_5$.
        \item $B$ locks his tokens on $Chain\text{-}B$ with $\mathcal{H}(s)$. This locks $P_b+p_a$ until timelock $t_2$.
        \item $A$ claims $B$'s principal $P_b$ and her premium $p_a$ on $Chain\text{-}B$ with \( s \). 
        \item $B$ uses the revealed preimage \( s \) to claim $A$'s principal $P_a$ and his premium $p_b$.
        \item If $A$ does not claim $B$'s tokens, $B$ can refund $P_b+p_a$ after \( t_2 \) expires.
        \item If $B$ does not claim $A$'s tokens, $A$ can refund $P_a+p_b$ after \( t_5 \) expires.
    \end{enumerate}

    Hedged Atomic Swap handles the two griefing cases mentioned in Section \ref{grief}. For the first case where $A$ locks $P_a$ before $t_4$ and $B$ leaves the exchange, $A$ can refund $P_a+p_b$ after $t_5$ because $B$ cannot refund his premium after $A$ locks his principal. Since $B$ has not locked his principal, $A$ can refund her premium $p_a$ after $t_1$. Thus, $B$ loses $p_b$ to $A$, disincentivizing him from leaving the exchange.
    For the second case, where $A$ refuses to claim after $A$ and $B$ lock their principal, $A$ and $B$ can be refunded after $t_5$ and $t_2$, respectively. Thus, $A$ gets $P_a+p_b$ and $B$ gets $P_b+p_a$ by refunding. To ensure that $A$ is disincentivized from leaving the exchange for this case, $p_a$ is kept larger than $p_b$. Thus, $A$ loses $p_a-p_b$ by leaving the exchange. However, this introduces the problem of griefing on premiums as the $B$ can leave the swap after $A$ locks her premium, and similarly, $A$ can leave the swap after $B$ locks his premium. This leads to the waiting until their respective refund timelock expires.

    \subsection{Grief-Free Atomic Swap}
    In 2022, Nadahalli \textit{et al.}~\cite{DBLP:conf/icbc2/NadahalliKW22} proposed Grief-free atomic swaps that guarantee no griefing while reducing the on-chain transactions to five. The idea is to combine the locking of the premium with the principal. However, it could not match the number of transactions in the TN-Swap as the combination was done only in one chain.

    \begin{figure}
        \centering
        \includegraphics[width=0.7\textwidth]{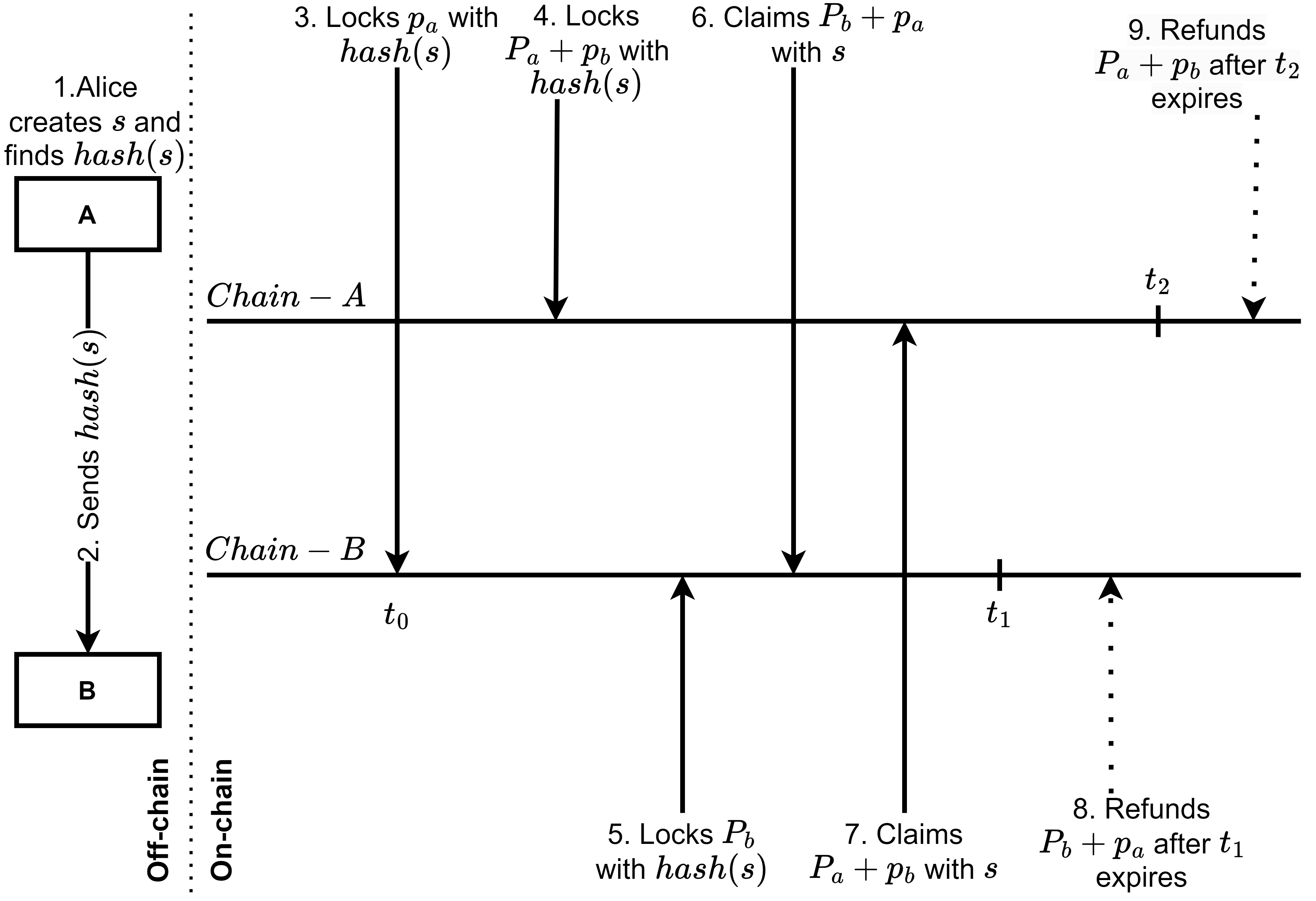}
        \caption{Grief-Free Atomic Swap Protocol}
        \label{fig:grieffreeswap}
    \end{figure}
    
    Again, consider the same setup as in the Section \ref{hedgedswap}. The following steps help achieve a grief-free atomic swap in five steps, as shown in Figure~\ref{fig:grieffreeswap}.
    \begin{enumerate}
        \item $A$ generates secret \( s \), finds $\mathcal{H}(s)$.
        \item $A$ sends $\mathcal{H}(s)$ to $B$.
        \item $A$ locks her premium $p_a$ with $\mathcal{H}(s)$ setting up a refund timelock of $t_1$ on $Chain\text{-}B$.
        \item $A$ and $B$ create a locking transaction $P_a+p_b$ and $A$ publishes it on $Chain\text{-}A$. This also sets up refund timelock $t_2$. If $B$ leaves before creating this, $A$ can refund $p_a$ by using $s$.
        \item $B$ locks his tokens with same $\mathcal{H}(s)$. This locks $P_b+p_a$ on $Chain\text{-}B$ until timelock $t_1$.
        \item $A$ claims $B$'s principal $P_b$ and her premium $p_a$ on $Chain\text{-}B$ with \( s \). 
        \item $B$ uses the revealed preimage \( s \) to claim $A$'s principal $P_a$ and his premium $p_b$.
        \item If $A$ does not claim $B$'s tokens, $B$ can refund $P_b+p_a$ after \( t_1 \) expires.
        \item If $B$ does not claim $A$'s tokens, $A$ can refund $P_a+p_b$ after \( t_2 \) expires.
    \end{enumerate}

    The GF-Swap handles the griefing case similarly to the XH-Swap; the difference is that it combines the premium with the principal and solves the griefing problem on premiums. Table~\ref{tab:atomic_swap_comparison} compares various atomic swap protocols along with 4-Swap. We now present the model and our assumptions for the 4-Swap protocol.

\begin{table}[ht]
\centering
\caption{Comparison of Atomic Swap Protocols. The column \textit{Txns} indicates the total number of on-chain transactions required by the protocol in the worst-case scenario.}
\begin{tabular}{|l|c|c|c|c|}
    \hline
    \textbf{Protocol} & \textbf{Griefing Resistance} & \textbf{Bribery Safety} & \textbf{Rev. Bribery Safety} & \textbf{Txns} \\
    \hline
    Tier-Nolan~\cite{tnswap} & No  & No  & No  & 4 \\
    Hedged~\cite{DBLP:conf/podc/XueH21}     & Yes & No  & No  & 6 \\
    Grief-Free~\cite{DBLP:conf/icbc2/NadahalliKW22} & Yes & No  & No  & 5 \\
    Rapidash~\cite{chung2022rapidash}   & No  & Yes & Yes & 6 \\
    4-Swap     & Yes & Yes & No  & 4 \\
    \hline
\end{tabular}
\label{tab:atomic_swap_comparison}
\end{table}

\section{System Model}
We assume two UTXO-based blockchains that allow multi-signature transactions. All the interacting parties, including miners, are assumed to be rational, and they try their best to increase their incentives. The transactions in both blockchains first go into a pool of unconfirmed transactions known as the mempool, from which miners then take some transactions to create a block published in the blockchain. Transactions that utilize the same UTXO (conflicting transactions) can also be published in the mempool, from which only one eventually gets into the blockchain, and the rest become invalid. We also assume synchronous communication and the presence of a global clock. A transaction may be active or inactive based on whether it can be added to a block immediately or not. Thus, the transactions that can be added to a block after some timelock are called inactive transactions; otherwise, if they can be immediately added to a block, they are called active transactions. Both blockchains allow the publishing of both active and inactive transactions, which are stored in the mempool. Additionally, a collision-resistant hash function $\mathcal{H}$ converts inputs of any length to constant-length outputs. Also, all the parties can see the state of both the blockchains along with their mempools. The 4-Swap protocol assumes that each party must hold some tokens on the chain where they want to receive the tokens to safeguard against potential attacks, as done in all related work~\cite{DBLP:conf/podc/XueH21, DBLP:conf/icbc2/NadahalliKW22}. Also, we assume that there is a significant difference in the principal and premium amount locked by the participating parties to accommodate fluctuations in the price of tokens. The premium amount is significantly higher than the transaction fees in both chains, and the value gained by exchanging the tokens is less than losing any of the premium. Moreover, the protocol does not consider the reverse bribery attacks i.e when the miners incentivizes the involved parties to publish conflicting transactions.

\subsection{UTXO Model}
\label{utxo}
We adopt the notation for UTXO-based blockchains from~\cite{9519487}, and use charts to visualize the process and clarify the transaction flow as shown in Figure~\ref{fig:transactionschema}. Transactions are represented as rounded rectangles. Transactions already confirmed on the blockchain are shown with double borders, while those not yet confirmed have single borders. Each transaction (rounded rectangle) includes one or more output boxes (with squared corners) showing the coins assigned to each output. Arrows indicate the conditions for spending them.

Typical conditions are abbreviated for simplicity. Most outputs require signatures from one or more parties, with public keys shown under the arrow. Additional conditions like timelocks are shown above the arrow. For example, a relative timelock, requiring that \( t \) rounds have passed since the transaction was published, is represented as \(+t\) above the arrow. Based on the number of rounds since the blockchain's creation, absolute timelocks are labelled \(>t\).

If the output's spending condition includes multiple options (a disjunction), the condition is written as \( \phi = \phi_1 \lor \dots \lor \phi_n \), where \( n \in \mathbb{N} \). The charts depict this using a diamond shape for the output, and each condition \( \phi_i \) is displayed on a separate arrow. It also uses $\lor$ notation on conditions to represent disjunction. In the case of conjunction, it is represented in the chart as \( \phi = \phi_1 \land \dots \land \phi_n \). Also, any transaction of the form $tx_{x}^y$ depicts action $y$ happening to funds of $x$. It does not depict who is publishing the transaction. Certain outputs can be spent by anyone labelled as ACS in Figure~\ref{fig:transactionschema} without the need for signatures. The grey square represents who can spend the output.

\begin{figure}[htbp]
    \centering
    \includegraphics[width=0.7\textwidth]{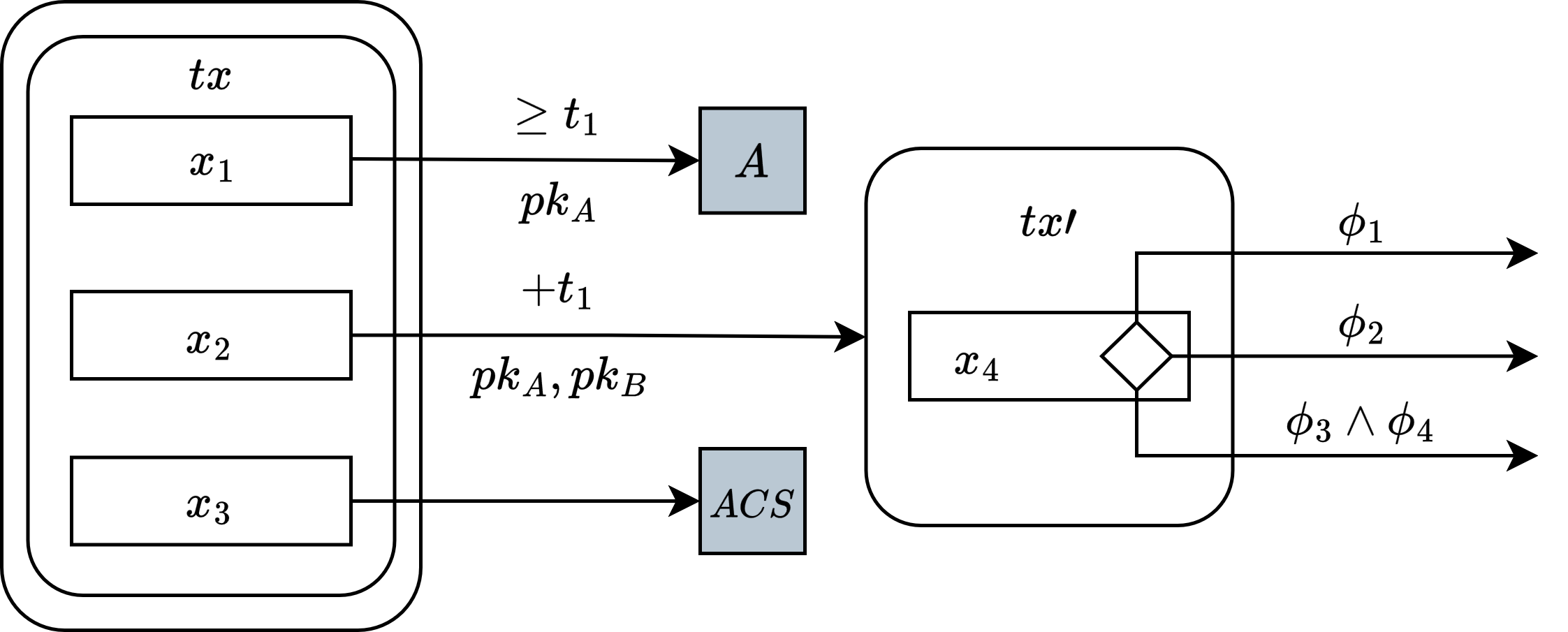 }
    \caption{Transaction Schema Representation}
    \label{fig:transactionschema}
\end{figure}

\section{4-Swap Protocol}

We present the development of the protocol through two versions, detailing the rationale behind each modification leading to the final 4-Swap protocol. The initial version (v1) focuses on facilitating grief-free atomic swaps that require only four on-chain transactions while making some assumptions. The next version then keeps the grief-free property and makes the protocol bribery safe within four steps. The notations used in the protocol are listed in Table \ref{tab:notation} for two parties $A$ and $B$ on two chains $Chain\text{-}A$ and $Chain\text{-}B$.

\begin{table}[ht]
    \centering
    \caption{Notation used in the protocol}
        \begin{tabular}{|c|c|}
            \hline
            \textbf{Notation}&\textbf{Description}\\
            \hline
            $\mathcal{H}({s})$ & Hash of some secret $s$\\
            $s_{m}$ & Main secret created by $B$\\
            $P_a$ & Principal amount $A$ wants to exchange on $Chain\text{-}A$\\
            $p_a$ & Griefing premium of $A$ on $Chain\text{-}B$\\
            $x_a$ & Bribery premium of $A$ on $Chain\text{-}A$\\
            $P_b$ & Principal amount $B$ wants to exchange on $Chain\text{-}B$\\
            $p_b$ & Griefing premium of $B$ on $Chain\text{-}A$\\
            $x_b$ & Bribery premium of $B$ on $Chain\text{-}B$\\
            $t_{i}$ & Timestamps in the chains\\
            $tx_x^y$ & Transaction that performs action $y$ on tokens of $x$\\
            $s_{r1}$ & Refund secret created by $A$\\
            $s_{r2}$ & Refund secret created by $B$\\
            $s_{br}$ & Bribery secret created by $A$\\
            $s_{e}$ & Early execution secret created by $A$\\
            \hline
        \end{tabular}
    \label{tab:notation}
\end{table}

\subsection{4S-v1}
In this initial version (v1), we relax the constraints typically associated with transaction processing by assuming that transactions are added to the blockchain immediately upon publication, with zero delay. This assumption is essential for considering conflicting transactions using the same UTXO. It ensures that the first active transaction put on the blockchain amongst many active or inactive conflicting transactions gets into it and can never be censored, eliminating the possibility of griefing. The transaction schema of $4S-v1$ is provided in Figure~\ref{4sv1_utxo}.

\begin{figure}
\centering
\includegraphics[width=0.7\textwidth]{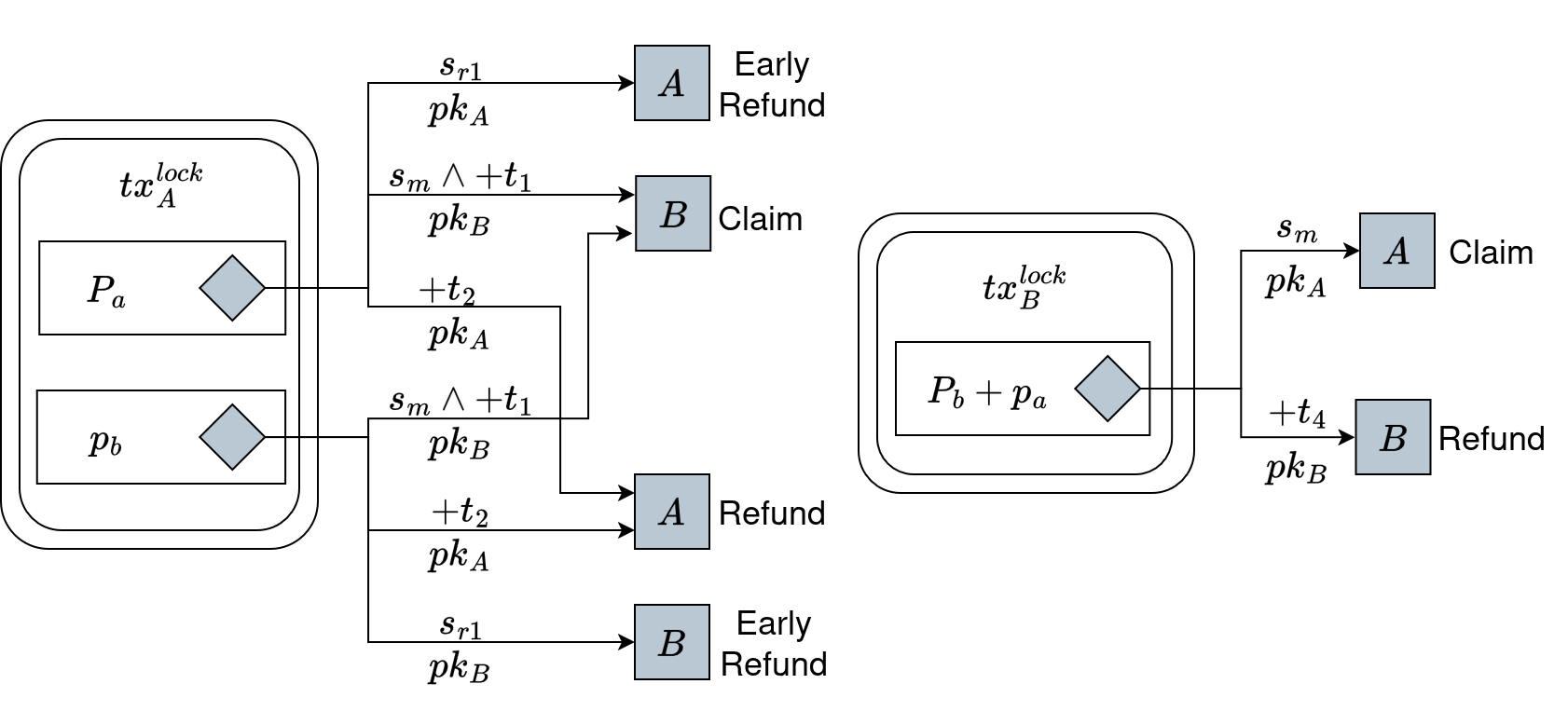}
\caption{Transaction Schema for $4S\text{-}v1$} \label{4sv1_utxo}
\end{figure}

Atomic swaps use griefing premiums to prevent griefing attacks, but locking these premiums separately increases the number of on-chain transactions~\cite{DBLP:conf/icbc2/NadahalliKW22}. To reduce this, we combine the premium with the principal amount, limiting the total number of on-chain transactions to four. The combination is done crosswise: A's premium is locked with B's principal, and B's premium is locked with A's principal, ensuring the premium is returned during the claim transaction. The premium structure is designed such that the party claiming first has to pay a higher griefing premium, as it can cause grief to the counterparty by not claiming. In the first case of griefing, if a party quits after the counterparty has already locked, the counterparty can refund its principal and the quitter's premium, thereby penalizing the quitter. However, in the second case, if a party quits before claiming (after both have locked), both parties can refund each other's premiums after the refund timelock. Thus, the first claimer has to lock a higher premium to disincentivize it from quitting the exchange.

We introduce cross-publishing of locking transactions and an option for an early refund for the party whose tokens get locked first. Cross-publishing means only B can publish A's locking transaction, and only A can publish B's. Cross-publishing helps to let the party know that their locking transaction is put on the chain by the counterparty, and they should do something based on it if they have the counterparty's locking transaction. Suppose B publishes A's locking transaction; if A has B's locking transaction, he should publish it; otherwise, he should immediately do an early refund on the locking transaction B put on the chain. If A makes an early refund of the principal amount, B should also be able to make the refund on the premium locked in the locking transaction; thus, another secret $s_{r1}$, which, when revealed during the early refund of A's principal amount, helps B to early refund his premium. $s_{r1}$ just indicates that $A$ is trying to refund his tokens. To ensure that A gets enough time to make an early refund, we introduce a new timelock $t_1$ as shown in Figure~\ref{4sv1_flow}, which helps to restrict B from claiming until the timelock expires. Cross-publishing also eliminates the first case of griefing, where a party leaves after the counterparty locks the tokens. If A leaves after B publishes A's locking transaction, B can claim A's principal after $t_1$. Also, suppose A publishes B's locking transaction first by moving away from the protocol. In that case, B can either put A's locking transaction or wait until refund timelock to get back the principal and A's griefing premium. Thus, the protocol is grief-free, as shown in Lemma ~\ref{griefinglemma}. We define the 4-Swap v1 protocol in three phases of setup, initiation and redeeming as in~\cite{DBLP:conf/sp/TsabaryYME21}.

The protocol 4S-v1 describes an atomic swap between participants $A$ and $B$, with $B$ initiating the process. The principal amounts are denoted as $P_a$ for $A$ and $P_b$ for $B$, while $p_a$ and $p_b$ represent the respective griefing premium amounts. $t_1$ is the timelock for the early refund on $Chain\text{-}A$ and $t_2$, $t_4$ are the refund timelocks for $Chain\text{-}A$ and $Chain\text{-}B$ respectively. The following are the steps:

\begin{enumerate}
    \item \textbf{Setup Phase}
    \begin{enumerate}
        \item $B$ generates a secret \( s_m \), finds $\mathcal{H}(s_m)$ and sends $\mathcal{H}(s_m)$, UTXO $p_b$ to $A$.
        \item $A$ generates a secret \( s_{r1} \), finds $\mathcal{H}(s_{r1})$, creates multi-sig transaction $tx_{A}^{lock}$ that locks $P_a+p_b$ with $\mathcal{H}(s_m)$ and sets up two refund timelocks $t_1$(early refund) and $t_2$. $s_{r1}$ is required if $A$ wants to refund the locked tokens early. $A$ then sends $tx_{lock,A}$ (only signed by $A$), $\mathcal{H}(s_r1)$ and UTXO $p_a$ to $B$.
        \item $B$ creates a multi-sig transaction $tx_{B}^{lock}$ that locks $P_b+p_a$ with $\mathcal{H}(s_m)$ and sets up refund timelock $t_4$. $B$ then sends $tx_{lock,B}$ (only signed by $B$) to $A$. Nothing has been published in either of the chain until now.
        
    \end{enumerate}
    \item \textbf{Initiation Phase}
    \begin{enumerate}
        \item $B$ publishes $tx_{A}^{lock}$ to $Chain\text{-}A$, which is immediately included in a block.
        \item $A$ then publishes $tx_{B}^{lock}$ to $Chain\text{-}B$ which is immediately included in a block.
    \end{enumerate}
    \item \textbf{Redeeming Phase}
    \begin{enumerate}
        \item If $A$ creates and publishes a transaction for early refunding $P_a$ using $s_{r1}$, B can also publish a transaction for early refunding $p_b$ using the revealed $s_{r1}$. 
        \item If $A$ does not make an early refund, $B$ waits for $t_1$ to expire and creates and publishes a transaction to claim $P_a$ and $p_b$.
        \item If $B$ publishes $tx_{A}^{claim}$, then $A$ also creates and publishes $tx_{B}^{claim}$ to claim $P_b$ and $p_a$
        \item If $B$ doesn't claim by $t_2$, $A$ creates and publishes a refund transaction to get $P_a + p_b$.
        \item If $A$ doesn't claim by $t_4$, $B$ creates and publishes a refund transaction to get $P_b + p_a$.
    \end{enumerate}
\end{enumerate}

\begin{figure}
\centering
\includegraphics[width=0.7\textwidth]{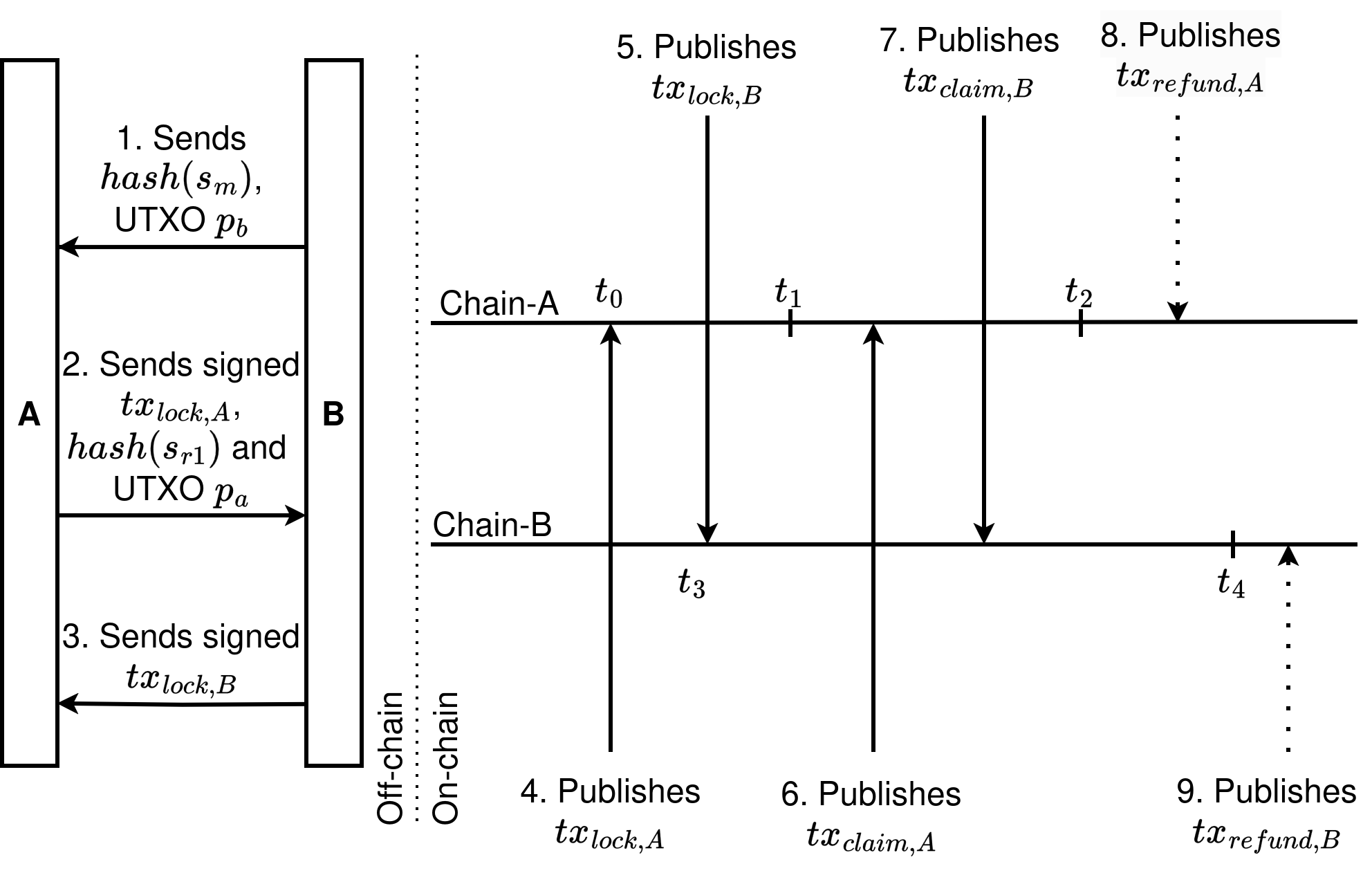}
\caption{Transaction flow for \(4S\text{-}v1\) and \(4S\), with \(4S\) differing slightly in off-chain interactions, as detailed in Section \ref{4Sprotocoldetails}.} \label{4sv1_flow}
\end{figure}

\begin{lemma}[Griefing Lemma]
    \label{griefinglemma}
    If any party abandons the protocol after tokens of the counterparty are locked on either of the two blockchains, the abandoning party incurs a penalty.
\end{lemma}
\begin{proof}
    Griefing can occur in two scenarios: first, if party \( B \) publishes party \( A \)'s locking transaction, but \( A \) abandons the protocol, and second, if \( B \) abandons the protocol after both locking transactions are confirmed. In the first case, after the timelock \( t_1 \), \( B \) can submit a claim transaction because it knows the secret \( s \), thereby retrieving the total amount \( P_a + p_b \), where \( P_a \) and \( p_b \) are the tokens locked by \( A \) and \( B \), respectively. Consequently, \( A \) incurs a loss of \( P_a \). In the second case, if \( B \) abandons the protocol after both locking transactions are confirmed, \( A \) can refund its locked tokens \( P_a + p_b \) after the timelock \( t_2 \), but \( A \) still incurs a penalty of \( p_a \), the tokens it initially locked. To discourage \( B \) from abandoning the protocol, the value of \( p_b \) is set to be greater than \( p_a \), ensuring that \( B \) faces a larger penalty in such scenarios. This mechanism ensures the protocol's resilience against griefing attacks.

\end{proof}

\subsection{4-Swap Protocol Details}
\label{4Sprotocoldetails}
In this section, we remove the assumption of zero delay in adding transactions to the blockchain after publication. In practice, transactions first enter the mempool and are confirmed only when miners include them in a block. To address this delay and ensure protocol security, we incorporate the concept of mutual destruction from~\cite{DBLP:conf/sp/TsabaryYME21}. This approach disincentivizes participants from publishing conflicting transactions(bribery scenarios) in the mempool. However, we do not address reverse bribery attacks, as seen in HTLCs~\cite{DBLP:conf/ndss/WadhwaSZN23} because it would increase the total on-chain transactions.

Mutual destruction works by embedding secrets within conflicting transactions. If these secrets enter the mempool together, any participant (typically miners) can use them by creating a transaction claiming the entire locked amount, making both conflicting transactions invalid. This mechanism discourages parties from attempting bribery. To implement this in the 4S protocol, we introduce the following two secrets in addition to $s_{r1}$ and the main secret $s_m$:
\begin{itemize}
    \item \( s_{br} \): $A$ creates $s_{br}$ for claiming $B$'s tokens. It helps the protocol to know that $A$ is claiming.
    \item \( s_{r2} \): $B$ creates $s_{r2}$ for refunding its tokens. It helps the protocol to know that $B$ is refunded.
\end{itemize}

Thus, whenever a pair of claim and refund transactions are published together in the blockchains, miner can create their slashing transaction to claim all the locked tokens. The transaction schema and all possible redeem paths for 4-Swap are shown in Figure~\ref{4s_utxo} and Figure~\ref{4s_utxoB}. This also covers the bribery case where an early refund transaction of $A$ is censored by claim transaction of $B$ ($p_b$ works as bribery premium of $B$).

\begin{figure}[htbp]
\includegraphics[width=0.9\textwidth]{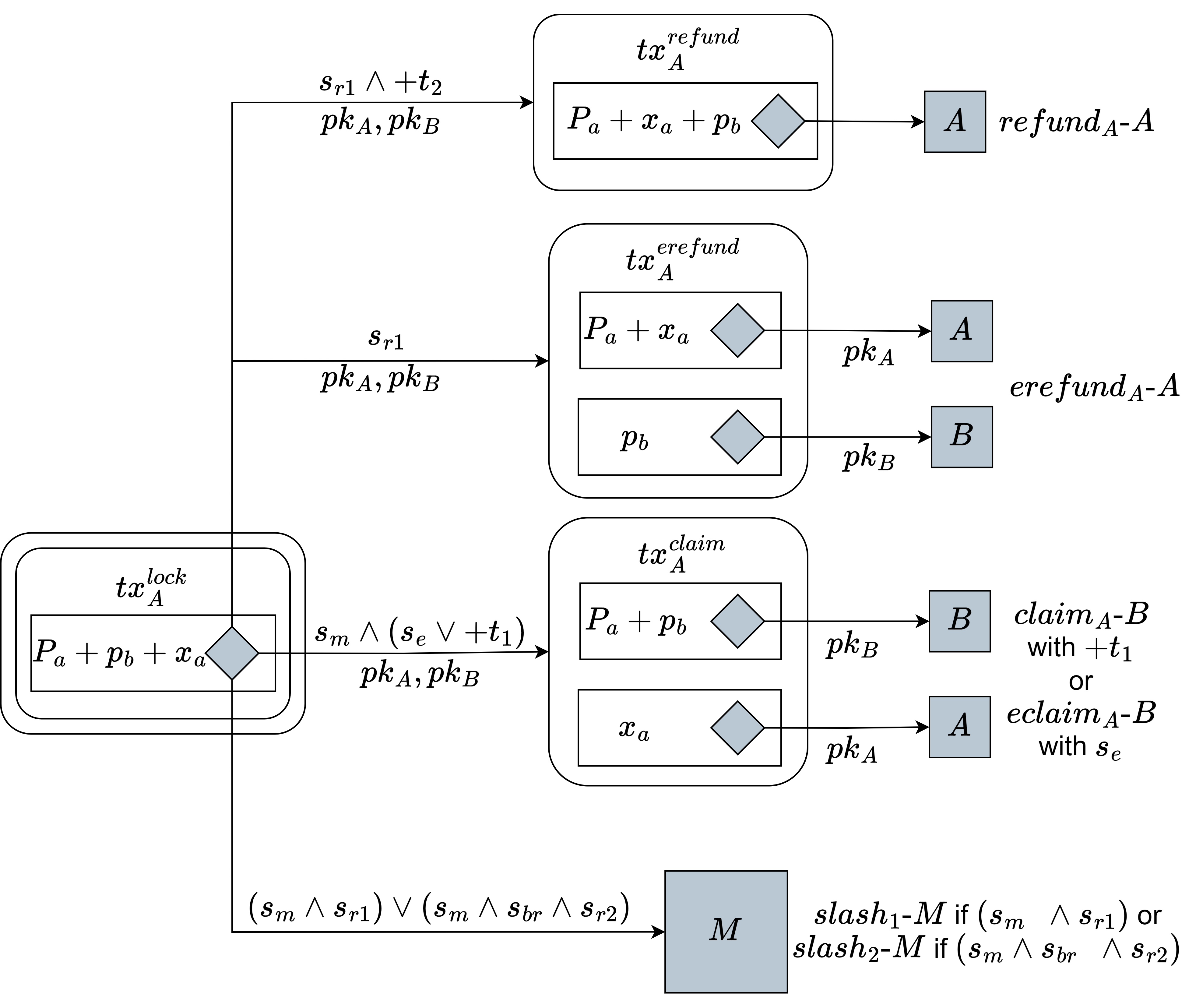}
\caption{Transaction Schema of ${4S}$ for $A$} \label{4s_utxo}
\end{figure}

\begin{figure}[htbp]
\includegraphics[width=0.9\textwidth]{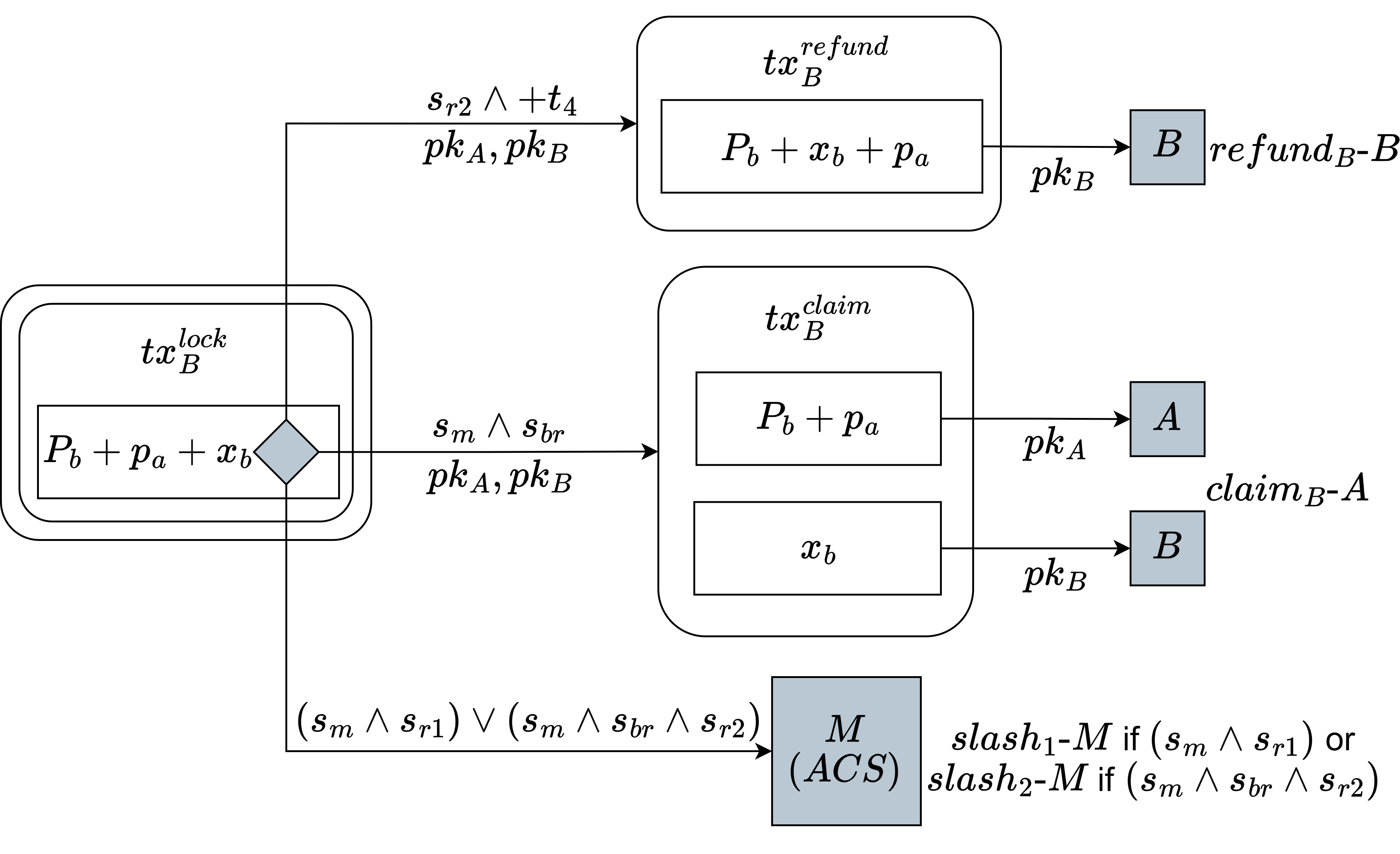}
\caption{Transaction Schema of ${4S}$ for $B$} \label{4s_utxoB}
\end{figure}

Additionally, in the setup shown in Figure~\ref{4sv1_flow}, even when both parties act honestly, $B$ must wait for \( t_1 \) to expire before claiming the funds. We introduce another secret, \( s_e \), generated by A for early execution. Once both locking transactions are confirmed on-chain, $A$ can share \( s_e \) with $B$, allowing $B$ to claim the funds without waiting for \( t_1 \) to expire.

However, griefing premiums alone are insufficient to address the bribery issue due to the following scenarios:
\begin{itemize}
    \item After $B$ successfully claims $A$'s tokens in \( Chain\text{-}A \), $B$ may attempt to bribe miners by publishing conflicting refund transactions, as $B$ no longer has anything at risk.
  
    \item $B$'s claim transaction remains in the mempool, while $A$ observes the corresponding secret from the mempool of \( Chain\text{-}A \). If $A$'s claim transaction is confirmed on \( Chain\text{-}B \) first, $A$ may attempt to bribe miners by submitting a conflicting refund transaction since $A$ similarly has nothing to lose.
\end{itemize}

The protocol introduces two bribery premiums to mitigate the abovementioned risks: \( x_a \) for party $A$ and \( x_b \) for party $B$. These premiums are not cross-locked, meaning each party must contribute an additional amount alongside their principal to be a bribery premium. Moreover, we incorporate static claim transactions in the swap. In the setup phase, both parties jointly create their claim transactions, which cannot be altered during the protocol's runtime. These static claim transactions facilitate the exchange of principal amounts while returning the bribery premiums within the same transaction. This approach addresses the bribery risks in the abovementioned scenarios. If a party attempts bribery, they risk losing their bribery premium if the counterparty cannot claim it due to the bribery attempt. We also create static early refund and refund transactions for party $A$ on $Chain\text{-}A$ and static refunds for $B$ on $Chain\text{-}B$ to reduce the overhead of multiple redeem paths. $slash_1$ and $slash_2$ in Figure~\ref{4s_utxo} and Figure~\ref{4s_utxoB} represent the redeem paths available to miners on individual locking transactions in case of bribery. The protocol flow is the same as Figure~\ref {4sv1_flow}, but just in off-chain interactions, the claim, refund and early refund transactions are created and signed before signing the locking transactions. When a party(miner) uses $s_m, s_{r1}$ to slash, it comes under $slash_1$, otherwise if it slashes by $s_m,s_{br}, s_{r2}$ it comes under $slash_2$. We now present the $4S$ protocol in three phases.

Protocol $4S$ defines an atomic swap between participants $A$ and $B$, initiated by $B$. Principal amounts are $P_a$ for $A$ and $P_b$ for $B$, with griefing premiums $p_a$ and $p_b$, and bribery premiums $x_a$ and $x_b$, respectively. Figure~\ref{4s_utxo} and Figure~\ref{4s_utxoB} shows the detailed construction of the two locking transactions and how the subsequent transactions will be created. The following are the required steps:

\begin{enumerate}
    \item \textbf{Setup Phase}
    \begin{enumerate}
        \item $B$ generates secrets $s_m$, $s_{r2}$, finds $\mathcal{H}(s_m), \mathcal{H}(s_{r2})$ and sends $\mathcal{H}(s_m)$, UTXO $p_b$ to $A$.
        \item $A$ generates secrets $s_{r1}$, $s_e$, $s_{br}$, finds $\mathcal{H}(s_{r1})$, $\mathcal{H}(s_{e})$, $\mathcal{H}(s_{br})$, creates multi-sig transaction $tx_{A}^{lock}$ that locks $P_a+p_b$ with all hashes with different conditions and sets up two refund timelocks $t_1$(early refund) and $t_2$. $s_{r1}$ is required if $A$ wants to refund the locked tokens early. $A$ also creates multi-signature transactions $tx_{A}^{claim}$, $tx_{A}^{erefund}$, $tx_{A}^{refund}$ using $tx_{A}^{lock}$ and signs it. The $tx_{A}^{claim}$, $tx_{A}^{erefund}$ and $tx_{A}^{refund}$ are completely signed before signing $tx_{A}^{lock}$. $A$ then sends $tx_{A}^{lock}$ (only signed by $A$), $tx_{A}^{claim}$, $\mathcal{H}(s_{r1})$, $\mathcal{H}(s_{e})$, $\mathcal{H}(s_{br})$ and UTXO $p_a$ to $B$.
        \item $B$ creates a multi-sig transaction $tx_{B}^{lock}$ that locks $P_b+p_a$ with all hashes with different conditions and sets up refund timelock $t_4$. $B$ also creates a multi-signature transaction $tx_{B}^{claim}$ and $tx_{B}^{refund}$  using $tx_{B}^{lock}$ and signs it. $B$ then sends the $tx_{B}^{lock}$ and $tx_{B}^{claim}$ to $A$. The $tx_{B}^{claim}$ and $tx_{B}^{refund}$ are completely signed before signing $tx_{B}^{lock}$. $B$ then sends $tx_{lock,B}$ (signed by $B$) to $A$. Nothing has been published in either chain until now.
        
    \end{enumerate}
    \item \textbf{Initiation Phase}
    \begin{enumerate}
        \item $B$ publishes $tx_{A}^{lock}$ to $Chain\text{-}A$.
        \item $A$ then publishes $tx_{B}^{lock}$ to $Chain\text{-}B$.
    \end{enumerate}
    \item \textbf{Redeeming Phase}
    \begin{enumerate}
        \item If $A$ and publishes $tx_{A}^{erefund}$ for early refunding $P_a+x_a$ using $s_{r1}$, $B$ also gets back $p_b$. 
        \item If $tx_{A}^{erefund}$ is not published, $B$ waits for $t_1$ to expire and publishes $tx_{A}^{claim}$ to claim $P_a$ and $p_b$. This also refunds $x_a$ to $A$.
        \item $B$ can publish without waiting if $A$ releases the early execution secret $s_e$ to $B$.
        \item If $B$ publishes $tx_{A}^{claim}$, then $A$ also creates and publishes $tx_{B}^{claim}$ to claim $P_b$ and $p_a$. this also refunds $x_b$ to $B$.
        \item If $B$ does not claim before $t_2$, $A$ creates and publishes $tx_{A}^{refund}$ to get $P_a$, $p_b$ and $x_a$.
        \item If $A$ does not claim before $t_4$, $B$ creates and publishes $tx_{B}^{refund}$ to get $P_b$, $p_a$ and $x_b$.
    \end{enumerate}
\end{enumerate}

Reducing the on-chain transactions to four helps in reducing the overall time taken by the exchange in both the worst case as well as for the best case. For the best case, the parties follow the protocol and the early execution secret $s_e$ is also shared. The protocol thus just requires four transactions, each followed by some confirmation delay. For the worst case, if party $B$ leaves after both the locking transactions have been published, both parties have to wait until the refund timelocks expire in both the chains to refund their tokens. The refund timelocks entered will be smaller than the refund timelocks in the earlier works~\cite{DBLP:conf/podc/XueH21, DBLP:conf/icbc2/NadahalliKW22} as it requires fewer on-chain transactions.

The security analysis of the protocol is detailed in Section~\ref{securityanalysis}. We now present a game-theoretical analysis of the 4-Swap protocol to ensure that the protocol remains in the equilibrium state in the presence of rational parties and miners

\section{Game Theoretical Analysis}
The protocol is modelled as an Extensive Form Perfect Information Game, where all parties know the latest states of blockchains and their mempools~\cite{aumayr2024securing}. The game consists of three phases: Setup, Claim, and Refund. The Setup phase covers the protocol's setup and locking steps, while the Claim and Refund phases focus on claiming and refunding locked funds.

The game involves three rational players: \(A\), \(B\), and the miner \(M\). The miner \(M\) is initially assumed to be the strongest and only miner on both blockchains, and is responsible for mining all blocks. Since \(M\) controls the blockchain, it can always accept bribes to censor transactions. However, we demonstrate that the protocol effectively incentivizes honest behaviour, ensuring that bribery is not rational for any party.

The game progresses in rounds. In each round, \(A\) and \(B\) choose their actions, after which the miner \(M\) confirms transactions and advances the game to the next round. The complete game tree is constructed by combining the trees for the Base, Claim, and Refund trees in Figure~\ref{fig:gamesetup}, Figure~\ref{fig:gameclaim} and Figure~\ref{fig:gamerefund}.

Initially, we analyze the protocol under the assumption of a single powerful miner, \(M\). However, we assume this miner $M$ does not intentionally delay the confirmation of transactions without bribery. We later show that this is true in the case of a multi-miner setup because of opportunity cost. We then generalize the model to a multi-miner setup across both chains. The analysis shows that the actions and utilities of \(A\) and \(B\) remain consistent even in the multi-miner setting, illustrating the protocol's robustness and ability to maintain incentives for honest participation under varying conditions.

\begin{definition}[Perfect Information Extensive Form Game]
A game with perfect information in extensive form can be described as a tree structure and is formally defined by the tuple 
\((N, H, P, A_i, u_i)\) where \(i \in N\):

\begin{itemize}
    \item $N$: A finite set of $n$ players, labeled as $N = \{1, 2, \dots, n\}$. Each decision point in the game that is not terminal is associated with a specific player $ I \in N$, responsible for choosing that stage.
    
    \item $H$: The collection of all possible action sequences, called histories, where each sequence $h$ leads to a specific node within the game tree. A subset $Z \subseteq H$ represents terminal histories, which correspond to the endpoints of every possible playthrough, also known as the tree's leaf nodes.
    
    \item $P$: A mapping called the player function, which determines the player assigned to decide each non-terminal history $h \in H \setminus Z$.
    
    \item $A_i$: A function that specifies the set of actions available to a player $i$ at a given history $h$, denoted as $A_i(h)$. This set represents the choices available to player \(i\) at that point in the game. Each branch extending from a node in the tree corresponds to one of these actions.
    
    \item $u_i$: A function representing the payoffs for each player $i$, defined as $u_i: Z \to \mathbb{R}$. This function assigns a real number to each terminal history $z \in Z$, which reflects the utility that player $i$ receives if the game ends at that particular terminal state.
    
\end{itemize}    
\end{definition}

The parties ${A, B, M}$ will resort to different actions at different game histories based on the utilities they will get from the actions. A strategy for a party is the sequence of actions it will take as the game proceeds based on the actions already taken by all the parties. The set of all such strategies of all the players is known as the strategy profile. We now formally present the definition of strategy profile.

\begin{definition}[Strategy Profile]
    In an extensive form game with perfect information, a strategy profile specifies the action $a \in A_i(h)$ chosen by each player $i \in N$ at every history $h$ when required to act. Specifically, for each player $i \in N$, a strategy $s_i$ is a function mapping the set of histories $H_i = \{ h \in H: P(h) = i \}$ to the set of actions $A_i$. This mapping satisfies $s_i(h) \in A_i(h)$ for all $h \in H_i$. A strategy profile is a collection of strategies for all players,  $s = (s_1, s_2, \dots, s_n)$.
\end{definition}

Since the game is a perfect information game, all the parties know the best action for the other parties. Since parties are rational, they will try to increase their utility in each subgame of the main game. Subgame Perfect Nash Equilibrium gives the path of the game from which the parties will not deviate, as it gives them the best utility. We now define the Subgame-Perfect Nash Equilibrium.

\begin{definition}[Subgame Perfect Nash Equlibrium]
    A strategy profile  $s^*$ $=$ $($$s_1^*$, $s_2^*$, $\dots$, $s_n^*$$)$ qualifies as a Subgame Perfect Nash Equilibrium (SPNE) if, in every subgame $G'$ of the game $G$, each player's strategy $s_i^*$ is the optimal response to the strategies of all other players. 
    
    More specifically, let $H'$ represent the set of all histories within the subgame $G'$. For each player $i$, the strategy $s_i^*$ is deemed optimal in $G'$ if the following condition holds: $
u_i(s_i^*, s_{-i}^*; h) \geq u_i(s_i, s_{-i}^*; h)$, for all strategies $s_i$ available to player $i$ in $G'$, and for every history $ h \in H'$. Here, $s_{-i}^*$ represents the strategies of all players other than $i$ in the SPNE, and $u_i$$($$s_i$, $s_{-i}^*$$;$ $h$$)$ denotes player $i$'s payoff when the strategy profile $(s_i, s_{-i}^*)$ is followed in the subgame starting at history $h$. A strategy profile is an SPNE if it ensures a Nash Equilibrium in every subgame.
\end{definition}

We first define the slashing lemma and then prove the game-theoretic safety of 4-Swap.

\begin{lemma}[Slashing Lemma]
    \label{slashinglemma}
    Slashing is an action available to the miner $M$, which can take away all the locked tokens of the parties $A$ and $B$. For any $h \in H_M$ of subgame $G'$ of our game where $M$ is miner, if $\exists$ $a=slash \in A_M(h)$, then
    $u_M(h, slash) > u_M(h, !slash)$.
\end{lemma}

\begin{proof}
The miner only gets the option to slash the locked tokens if there is a bribery scenario of censoring a current transaction by a conflicting transaction that becomes valid after a certain timelock. The miners will only censor the transaction if they get enough bribes from the bribing party to censor it. Assuming the lemma does not hold, at least one other action exists for the miner $M$, which gives it a better utility than slashing. Therefore, the bribe amount should be more than the tokens locked in the chains. Assuming the parties are rational, such a bribe will not be published. Thus, no action can provide a better utility than slashing, and the assumption that the lemma does not hold is false. Hence, rational miners will always choose to slash both parties' locked tokens.
\end{proof}

\subsection{Game Trees}
The game proceeds in rounds with parties putting all their available choices before a miner comes to confirm some transaction, and then the next round starts. The game is divided into three parts for better representation. The base tree is the starting point of the tree and starts with $A$ sharing the $tx_{lock}^A$ with $B$ or not. This tree gives information about different scenarios of early refunds and reflects that the best strategy is not to deviate from the honest path. The claim tree assumes both parties have locked their tokens in their chains. The claim tree starts with $B$, choosing whether to publish the claim transaction or not. The refund tree focuses on parties wanting to return their tokens without claiming. In any of the trees, wherever a miner gets the option to slash, it chooses that, and thus, all the other available options are not shown in the tree to reduce its size. The nodes in the tree represent the party that takes the action, and the action taken by that party is represented by the arrow that takes it to another node, which takes the next action. After all the transitions are finished, the leaf represents the utilities of the parties $A$, $B$ and $M$. In Figure~\ref {fig:gamesetup}, Figure~\ref {fig:gameclaim} and Figure~\ref {fig:gamerefund}, the utilities of all the parties in different game paths are represented at the leaf nodes. The utilities in $Chain\text{-}A$ and $Chain\text{-}B$ are represented with a superscript of $1$ and $2$, respectively, and $f$ represents the transaction fees the miner gets on confirming some transaction. The claim tree and refund tree are together called as redeem tree. The game trees do not directly represent the opportunity cost, but if two paths give the same utility, the path with the least delay is preferred because of the lesser loss of opportunity cost. Ideally, $B$ would prefer $P_a$ over $P_b$ and vice-versa for $A$ as they are interested in the exchange. In case of refund this interest of parties changes and thus the difference in griefing premium ($p_b - p_a$) helps in dis-incentivizing the parties from leaving.

\begin{figure}[h]
    \centering
    \includegraphics[width=\textwidth]{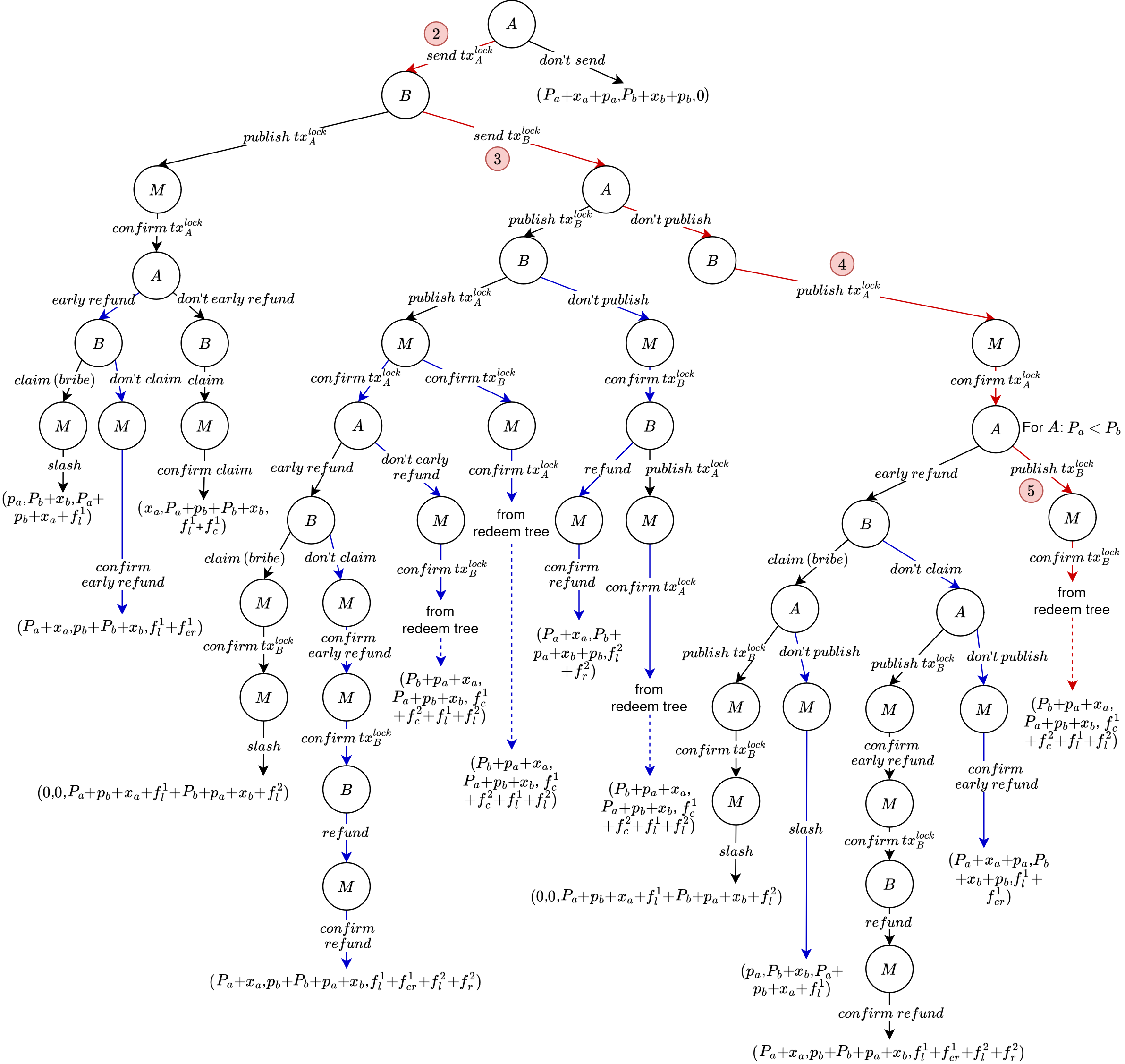}
    \caption{SPNE - Base Tree: Blue lines show sub-game optimal strategies; red lines show main game optimal strategies. Red path numbers correspond to transaction steps in Figure~\ref{4sv1_flow}}
    \label{fig:gamesetup}
\end{figure}

\subsubsection{Base Tree}
The base tree begins with $A$ deciding whether to send her locking transaction $tx_A^{lock}$ to $B$. If not, the protocol ends without affecting either party's tokens. If sent, $B$ can either publish $tx_A^{lock}$ or respond with $tx_B^{lock}$. If $B$ only publishes $tx_A^{lock}$, miners ($M$) confirm it and earn a fee $f_l^1$, and $A$ then chooses whether to publish an early refund. If $A$ does so and $B$ bribes miners to censor it, $M$ slashes all locked tokens. Otherwise, $A$ is refunded and $M$ earns $f_{er}^1$. If $A$ doesn't refund early, $B$ can claim the funds after a timelock, gaining $P_a$, while $A$ loses.

If $B$ sends $tx_B^{lock}$, $A$ then decides whether to publish it. If she does and $B$ withholds $tx_A^{lock}$, $M$ confirms $tx_B^{lock}$ and gains $f_l^2$. $B$ can now either refund, gaining both his funds and $A$'s griefing premium, or proceed with the exchange. By backward induction, refunding yields higher utility. If both publish their locking transactions, miners confirm one first, and parties act accordingly. In any bribery attempt, miners benefit from slashing, and both parties incur losses.

Transaction fees $f_{er}^1$, $f_r^1$ are for early and regular refunds on $Chain$-$A$, and $f_l^2$, $f_r^2$ are for locking and refunding on $Chain$-$B$. The redeem tree shows equilibrium utilities after locking, further explored in the claim tree. If parties exchange locking transactions and $A$ publishes first, backward induction shows $B$ prefers refunding to gain the griefing premium.

Finally, if both parties exchange transactions and $B$ publishes $tx_A^{lock}$ first (confirmed with fee $f_l^2$), $A$ can either early refund, publish $tx_B^{lock}$ or do nothing(not shown in Figure~\ref{fig:gamesetup}). Since early refunding results in fee loss, backward induction favours continuing the exchange by publishing $tx_B^{lock}$. The path where $A$ does not do anything will lead to $B$ claiming the locked assets as it knows the main secret.

\subsubsection{Claim Tree}
The claim tree represents the subgame where both parties have published and locked their tokens on the chain. It starts with $B$ having the option to publish a claim transaction or leave the exchange. Leaving the exchange means eventual refining by both parties, which we will discuss in the refund tree. If $B$ publishes a claim transaction revealing the secret $s_m$, $A$ can publish the claim transaction. In subsequent sub-games, whenever a party $A$ or $B$ gets an option to bribe by submitting the refund transaction on the chain, the best strategy for the miners is to slash; thus, all rational parties will refrain from doing so. Thus, the best strategy is to claim each other's tokens without refunding after publishing the claim transactions. $f_c^1$ represents the transaction fees of claim transaction on $Chain\text{-}A$ and $f_c^2$, $f_r^2$ represents claim and refund transaction fees on $Chain\text{-}B$.

\begin{figure}[htbp]
    \centering
    \includegraphics[width=\textwidth]{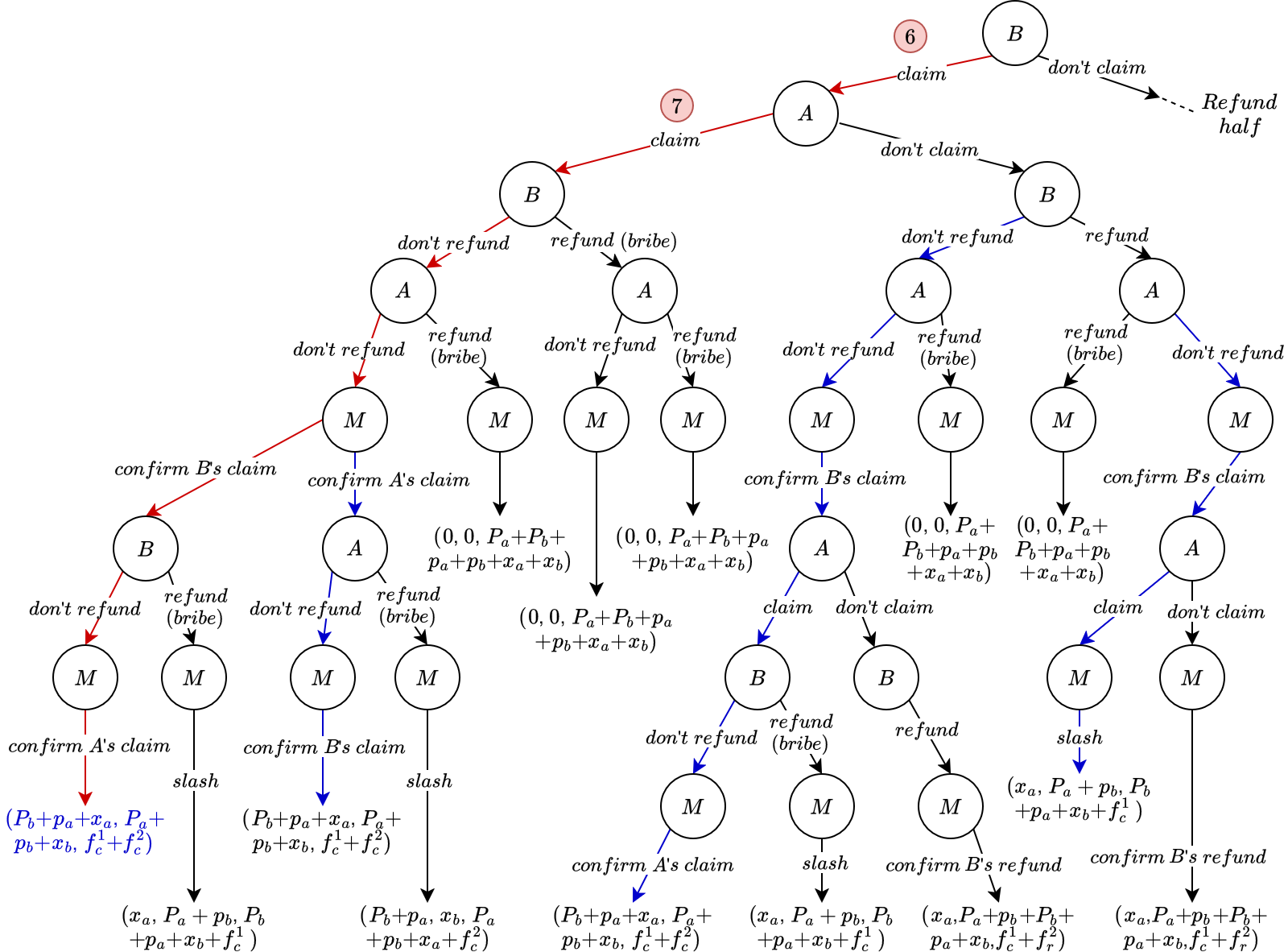}
    \caption{SPNE - Claim Half}
    \label{fig:gameclaim}
\end{figure}

\subsubsection{Refund Tree}
When $B$ does not publish a claim transaction, $A$ either publishes a refund transaction or does not. If she does not publish the refund transaction, $B$ will publish his refund transaction, which is eventually confirmed by miners by gaining $f_r^2$. After this, $B$ can still try to claim $A$'s transaction by publishing a claim, which can be contested by $A$ by publishing a refund, but the miners' best strategy would be to slash.
If $A$ publishes a refund after $B$ does not publish a claim transaction, the best strategy for $B$ is to also submit a refund transaction and avoid slashing actions for the miner. Thus, using backward induction, we can see from the refund tree that the best strategy is when both parties publish refund transactions without publishing their claim transactions. 

\begin{figure}[htbp]
    \centering
    \includegraphics[width=0.70\textwidth]{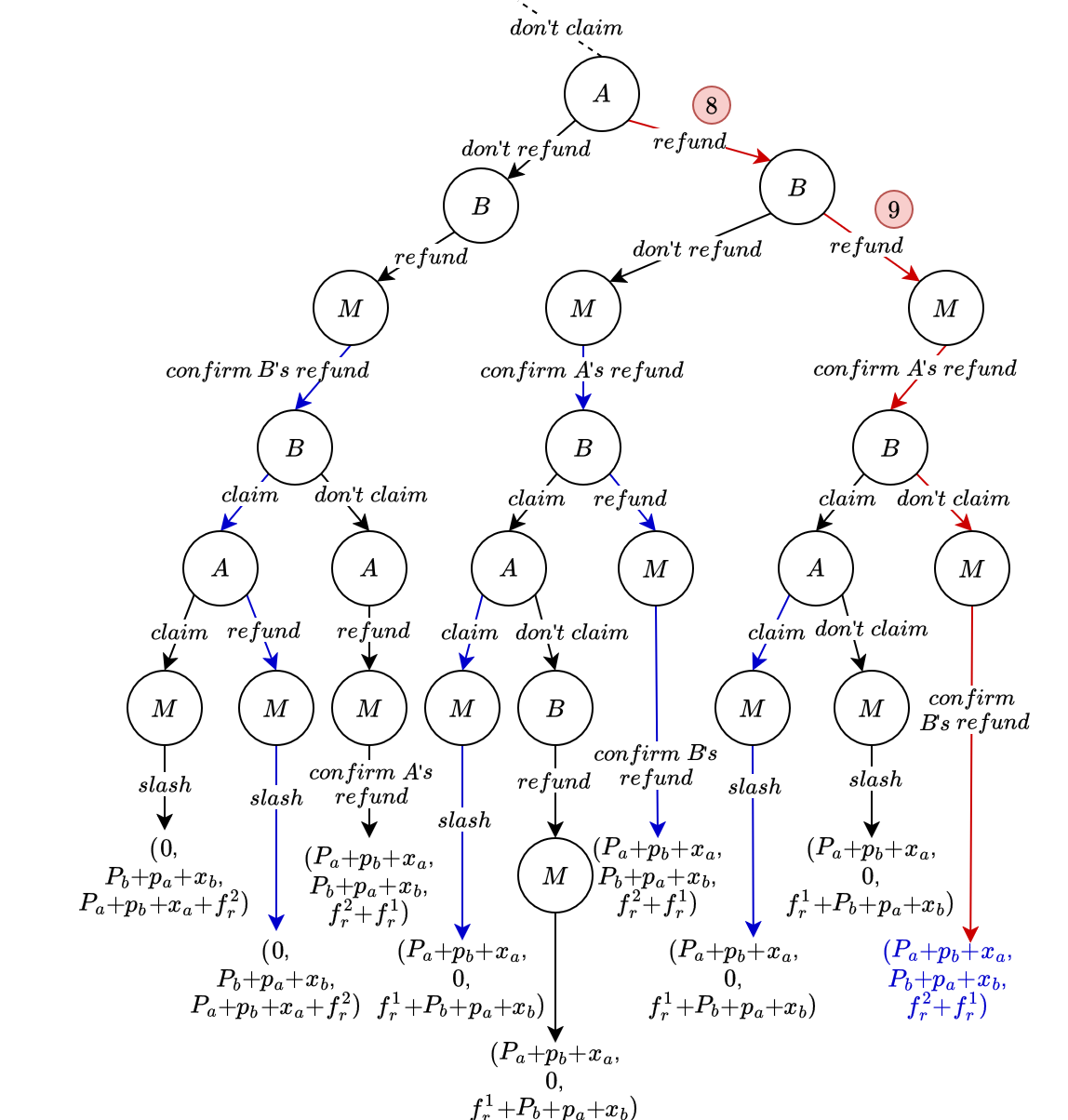}
    \caption{SPNE - Refund Half}
    \label{fig:gamerefund}
\end{figure}

To prove that our protocol is game-theoretically safe, we need to show that the protocol holds the Subgame Perfect Nash Equilibrium at the honest(expected) path in the game tree.

\begin{theorem}[SPNE]
    \label{equilibrium}
    The Subgame Perfect Nash Equilibrium for our game is the strategy profile $s^*(A, B, M) = $
    $($$send$ $tx_A^{lock}$, $don't$ $publish$, $publish$ $tx_B^{lock}$$)$, 
    $($$send$ $tx_B^{lock}$, $publish$ $tx_A^{lock}$, $claim$, $don't$ $refund$, $don't$ $refund$$)$, 
    $(($$confirm$ $tx_A^{lock}$, $confirm$ $tx_B^{lock}$, $confirm$ $B's$ $claim$, $confirm$ $A's$ $claim$$)$ $or$ $($
    $confirm$ $tx_A^{lock}$, $confirm$ $tx_B^{lock}$, $confirm$ $A's$ $claim$, $confirm$ $B's$ $claim$$))$.

\end{theorem}

\begin{proof}
    Lemma ~\ref{slashinglemma} ensures parties will never try to bribe the miners because of slashing. We apply backward induction at every subgame of the setup, claim, and refund phases of the game tree until we get the SPNE for the entire game. Figure~\ref {fig:gamesetup}, Figure~\ref {fig:gameclaim} and Figure~\ref {fig:gamerefund} show the SPNE of the entire 4-Swap game. Thus, parties can never get more utility than the prescribed(honest) strategy.
 \end{proof}

 Theorem ~\ref{equilibrium} shows that the honest course of action is the SPNE of our game. We now show that the strongest miner $M$ can be generalized to a multi-miner setup in both chains while ensuring the protocol's safety.

 \begin{lemma}[Miner's Dilemma]
    \label{minerdilemma}
    In the case of a multi-miner setup, a miner does not wait for the parties to publish conflicting transactions for slashing opportunities.
\end{lemma}

\begin{proof}
Lemma \ref{slashinglemma} proves that slashing gives the maximum utility to the miners. However, in the presence of rational parties, parties know that slashing is the best action for the miners, and thus, they will refrain from publishing conflicting transactions. Thus, the miner will lose the opportunity cost to some other miner by waiting for the conflicting transaction and hence does not wait for the same.
\end{proof}

 \begin{theorem}
    In the 4-Swap protocol, rational parties \(A\) and \(B\) maximize their utilities by following honest actions, even in a multi-miner setup.
 \end{theorem}
 \begin{proof} 
    Replacing a single all-powerful miner with arbitrary miners in a multi-miner setup yields identical utilities for parties $A$ and $B$ at all game tree leaf nodes. Lemma~\ref{slashinglemma} shows that the slashing mechanism deters any miner, regardless of hashing power, from delaying transactions due to bribery. Rational miners thus enforce the slashing mechanism when applicable. In non-bribery scenarios, miners can increase their rewards by confirming the only valid transaction (on any particular chain). Thus, arbitrary miners make the same decisions as an all-powerful miner. While the rewards in the game tree may be distributed among multiple miners in a multi-miner setup, this distribution does not affect the utilities of \(A\) and \(B\). Consequently, the rational strategies for \(A\) and \(B\) remain aligned with honest behaviour, preserving the protocol's game-theoretic safety in a multi-miner environment.
 \end{proof}

\section{Discussion}
To make the 4-Swap protocol reverse bribery safe along with griefing safety, we can leverage the idea of splitting the refund transaction into two steps, such that whenever a refund transaction is published, there is sufficient time for miners to claim a portion and burn the remaining tokens~\cite{DBLP:conf/ndss/WadhwaSZN23, chung2022rapidash}. Splitting the transaction in two steps helps in providing a punishing window, thereby solving the reverse bribery. In 4-Swap there are three bribery cases: early refund of $A$ getting censored by claim of $B$, claim of $B$ getting censored by refund of $A$ and claim of $A$ getting censored by refund of $B$. To handle all these cases, all the transactions that cause conflicts have to be split up. Thus, the claim of $B$ and the two main refund transactions also have to be split up into two transactions. Thus for the worst case where both parties refund, we will have six transactions and in best case where parties exchange their tokens, we will have five transactions. This still leads to an open question whether we can extend this construction to be reverse bribery safe in just four transactions.

 \section{Conclusion}
\label{conclusions}
We propose 4-Swap, the first atomic swap protocol to achieve both grief-free and bribery-resistant atomic swaps using only four steps. To our knowledge, this is the first work that combines principals and premiums in both chains while matching the number of transactions in TN-Swap, leading to faster execution of the protocol compared to previous grief-free atomic swaps. A game-theoretical analysis of the protocol shows that the honest execution path constitutes a stable equilibrium, and thus satisfies the goals of being grief-free and safe against bribery attacks.

\bibliography{lipics-v2021-sample-article}

\appendix

\section{Security Analysis}
\label{securityanalysis}
We use the notion of UC security~\cite{DBLP:conf/tcc/CanettiDPW07} for the three phases of setup, initiation and redeeming phases of $\Pi_{4S}$ based on~\cite{DBLP:conf/sp/TsabaryYME21}. The UC is used to prove the security for the three phases, and then put forward the conditions that the parties require to redeem the locked tokens. We now define our formal protocol $\Pi_{4S}$ in Figure~\ref {formalprotocol4s}.

\begin{figure}
\centering
\fbox{ 
    \begin{minipage}{\columnwidth}
    Protocol $\Pi_{4S}$ defines an atomic swap between participants $A$ and $B$, initiated by $B$. Deposits are $P_a$ for $A$ and $P_b$ for $B$, with griefing premiums $p_a$ and $p_b$, and bribery premiums $x_a$ and $x_b$, respectively.
    \\
    \rule{\textwidth}{0.4pt}
    \begin{center} \textbf{Setup Phase} \end{center}
    $B$ draws $s_m \stackrel{\text{R}}{\leftarrow} \{0, 1\}^{\mu}$, $s_{r2} \stackrel{\text{R}}{\leftarrow} \{0, 1\}^{\mu}$  and sets $dig_{m} \leftarrow H(s_m)$, $dig_{r2} \leftarrow H(s_{r2})$. Then $B$ sends $dig_{m}$ along with $UTXO-p_b$ to $A$. $A$ then draws $s_{r1} \stackrel{\text{R}}{\leftarrow} \{0, 1\}^{\mu}$, $s_{e} \stackrel{\text{R}}{\leftarrow} \{0, 1\}^{\mu}$,  $s_{br} \stackrel{\text{R}}{\leftarrow} \{0, 1\}^{\mu}$ and sets $dig_{r1} \leftarrow H(s_{r1})$, $dig_{e} \leftarrow H(s_{e})$, $dig_{br} \leftarrow H(s_{br})$. $A$ then creates $tx_{A}^{lock}$ that initiates swap on $Chain\text{-}A$ with $dig_{m}$, $dig_{e}$, $dig_{r1}$, $dig_{r2}$ $t_1$, $t_2$, $P_a$, $x_a$ and $p_b$ as parameters. $A$ also creates multi-signature transactions $tx_{A}^{claim}$, $tx_{A}^{erefund}$, $tx_{A}^{refund}$ using $tx_{A}^{lock}$ and signs it. The $tx_{A}^{claim}$, $tx_{A}^{erefund}$ and $tx_{A}^{refund}$ are completely signed before signing $tx_{A}^{lock}$.
    $A$ then sends the $tx_{A}^{lock}$, $tx_{A}^{claim}$, $dig_{r1}$, $dig_{e}$, $dig_{br}$ and $UTXO-p_a$ to $B$. $tx_{A}^{lock}$ is not published yet. $B$ then creates $tx_{B}^{lock}$ that initiates swap on $Chain\text{-}B$ with $dig_{m}$, $dig_{r2}$, $dig_{r1}$, $dig_{br}$, $t_4$, $P_b$, $x_b$ and $p_a$ as parameters. $B$ also creates a multi-signature transaction $tx_{B}^{claim}$ and $tx_{B}^{refund}$  using $tx_{B}^{lock}$ and signs it. $B$ then sends the $tx_{B}^{lock}$ and $tx_{B}^{claim}$ to $A$. The $tx_{B}^{claim}$ and $tx_{B}^{refund}$ are completely signed before signing $tx_{B}^{lock}$. $tx_{B}^{lock}$ is not published yet.
    \\
    \rule{\textwidth}{0.4pt}
    \begin{center} \textbf{Initiation Phase} \end{center}
\( B \) publishes \( tx_{A}^{lock} \) to the mempool of \( Chain\text{-}A\), and it is subsequently included in a block. Following this, \( A \) publishes \( tx_{B}^{lock} \) to the mempool of \( Chain\text{-}B\), which is also eventually confirmed in a block.
    \\
    \rule{\textwidth}{0.4pt}
    \begin{center} \textbf{Redeeming Phase} \end{center}
    $A$ can send $s_e$ to $B$ after $tx_{B}^{lock}$ gets confirmed for early claim without waiting until $t_1$.
    If $A$ creates and publishes $tx_{A}^{erefund}$ for early refunding $P_a+x_a$ using $s_{r1}$, $B$ also gets back $p_b$. 
    If $tx_{A}^{erefund}$ is not published, $B$ waits for $t_1$ to expire (if he does not have $s_e$) and publishes $tx_{A}^{claim}$ to claim $P_a+p_b$ and $x_a$ is returned to $A$. If $B$ publishes $tx_{A}^{claim}$, then $A$ also and publishes $tx_{B}^{claim}$ to claim $P_b+p_a$ and $x_b$ is returned. If $B$ does not claim before $t_2$, $A$ publishes $tx_{A}^{refund}$ to get $P_a+x_a+p_b$. If $A$ does not claim before $t_4$, $B$ publishes $tx_{B}^{refund}$ to get $P_b+x_b+p_a$.
    \end{minipage}
}
\caption{Protocol $\Pi_{4S}$}
\label{formalprotocol4s}
\end{figure}

\label{ucframework}
We first relax our protocol by considering all conflicting and inactive transactions that become valid after timelocks as valid transactions. Thus, we remove the time parameter from the redeeming conditions and then define $rPredicate$ as a function of party, path, and preimages known. $\mathcal{Q}$ represents all the possible paths in $rPredicate()$.

\[
\begin{aligned}
& rPredicate(path, P, h_m, h_{br}, h_e, h_{r1}, h_{r2}) = \\
& \begin{cases}
    (P = A) \land h_m \land h_{br} & \text{if } path = claim_B\text{-}A, \\
    (P = B) \land h_m & \text{if } path = claim_A\text{-}B, \\
    (P = B) \land h_m \land h_e & \text{if } path = eclaim_A\text{-}B, \\
    (P = A) \land h_{r1} & \text{if } path = refund_A\text{-}A, \\
    (P = B) \land h_{r2} & \text{if } path = refund_B\text{-}B, \\
    h_m \land h_{r1} & \text{if } path \in 
    \begin{aligned}[t]
        & \{ slash_{1,A}\text{-}M, slash_{1,B}\text{-}M \},
    \end{aligned} \\
    h_m \land h_{br} \land h_{r2} & \text{if } path \in 
    \begin{aligned}[t]
        & \{ slash_{2,A}\text{-}M, slash_{2,B}\text{-}M\}
    \end{aligned}
\end{cases}
\end{aligned}
\]

We develop an ideal functionality known as $\mathcal{G}_{mbp}$, which simulates the blockchain and mempool projection of two contracts across different blockchains. Also, we model hash function$(\mathcal{H})$ as the random oracle~\cite{canetti2014practical}. We simplify the model by assuming there are authenticated channels between the parties and the functionalities, eliminating the need to address digital signatures. Then, we define our relaxed formal protocol $\Pi_{r4S}$ under the hybrid world with $\mathcal{G}_{mbp}$ and $\mathcal{H}$. The difference between the $\Pi_{4S}$ and $\Pi_{r4S}$ is that the relaxed version ignores inactive, contradicting transactions and considers the miner part of the system. Now, we put forward the lemma describing all valid redeeming paths, the proof of which requires ideal functionality $\mathcal{F}_{r4S}$ and a simulator $Sim$.

\begin{lemma}
    \label{mainlemma}
    If the setup and initiation are completed according to $\Pi_{r4S}$, the variables $pub_m$, $pub_{r1}$, and $pub_{r2}$ indicate whether the preimages are stored in $\mathcal{G}_{mbp}$, while the variable $shared$ shows whether party $A$ has shared the early claim secret with party $B$. Initially, all indicators are set to 0. Under these conditions, the parties $A$, $B$, and miners $M_1$, $M_2$ can execute the following redeeming transactions:

    \begin{itemize}
        \item Party $B$ can initiate a valid redeeming transaction using $eclaim_A\text{-}B$ if either $shared=1$ or if using $claim_A\text{-}B$. By doing so, it updates $pub_m$ to 1. Additionally, $B$ has the option to initiate a valid redeeming transaction called $refund_B\text{-}B$, which will set $pub_{r2}$ to 1.
        \item Party $A$ can publish a valid redeeming transaction with either $erefund_A\text{-}A$ or $refund_A\text{-}A$, which will activate $pub_{r1}$. If $pub_m$ is set to 1, $A$ can also publish $claim_B\text{-}A$.
        \item If both $pub_m$ and $pub_{r1}$ are true, miners $M_1$ and $M_2$ can proceed with valid redeeming transactions in $slash_1$.
        \item If $pub_m$, $pub_{br}$, and $pub_{r2}$ are all true, miners $M_1$ and $M_2$ can perform valid redeeming transactions in $slash_2$.
    \end{itemize}
\end{lemma}

Lemma~\ref{mainlemma} is valid if there exists an ideal functionality, denoted as $\mathcal{F}_{r4S}$, which represents a relaxed 4-Swap ideal functionality in the ideal world. This functionality ensures that for any probabilistic polynomial-time (PPT) adversary $Adv$, there exists a PPT simulator $Sim$ such that no PPT environment $\mathcal{Z}$ can computationally distinguish between the execution of the protocol $\Pi_{r4S}$ with $Adv$ in the hybrid world and the execution of $\mathcal{F}_{r4S}$ in the ideal world with the $Sim$.

The ideal functionality \( \mathcal{F}_{r4S} \) manages the setup, initiation, and publication of redemption transactions for a session ID (\(sid\)). It interacts with parties \(A\), \(B\), \(M_1\), \(M_2\), and the simulator \(Sim\), maintaining several status indicators: \(setup_{B}^{r4S}\), \(setup_{A}^{r4S}\), \(setup^{r4S}\), \(published_{B}^{r4S}\), \(published_{A}^{r4S}\), \(init_{B}^{r4S}\), \(init_{A}^{r4S}\), \(shared^{r4S}\), \(pub_{m}^{r4S}\), \(pub_{r1}^{r4S}\), \(pub_{r2}^{r4S}\), \(pub_{e}^{r4S}\), and \(pub_{br}^{r4S}\), all initially set to 0. Upon receiving setup or publication requests from the parties, it updates these indicators and leaks relevant information to \(Sim\). When redeem requests are made, \( \mathcal{F}_{r4S} \) evaluates the requests based on the current state of these indicators using the predicate function \(rPredicate\) and returns the result to the requesting party, ensuring proper protocol execution and interaction. Fig~\ref{fig:idealfunctionality4S1} shows the complete ideal functionality $\mathcal{F}_{r4S}$.

\begin{figure}
\centering
\fbox{ 
    \begin{minipage}{\columnwidth}
    In the ideal world, the functionality \( \mathcal{F}_{r4S} \) is responsible for managing the setup, initiation, and publication of the redemption transaction for a specific session ID (\(sid\)). It interacts with parties \(A\), \(B\), \(M_1\), $M_2$, and the simulator \(Sim\), overseeing the proper execution of the protocol. Internally, it tracks the process through several status indicators, including \(setup_{B}^{r4S}\), \(setup_{A}^{r4S}\), \(setup^{r4S}\), \(published_{B}^{r4S}\), \(published_{A}^{r4S}\), \(init_{B}^{r4S}\), \(init_{A}^{r4S}\), \(shared^{r4S}\), \(pub_{m}^{r4S}\), \(pub_{r1}^{r4S}\), \(pub_{r2}^{r4S}\), \(pub_{e}^{r4S}\), and \(pub_{br}^{r4S}\), all of which are initially set to 0.
    \begin{itemize}
        \item Upon receiving input \((setup\text{-}B, sid)\) from $B$, if \(setup_{B}^{r4S} = 0\), set \(setup_{B}^{r4S} \leftarrow 1\) and leak \((setup\text{-}B, sid)\) to \(Sim\).
        \item Upon receiving input \((setup\text{-}A, sid)\) from $A$, if $setup_{A}^{r4S} = 0$, set $setup_{A}^{r4S} \leftarrow 1$ and leak \((setup\text{-}A, sid)\) to $Sim$.
        \item Upon receiving input \((setup_{2}\text{-}B, sid)\) from $B$ if $setup_{A}^{r4S} = 1$, set $setup^{r4S} \leftarrow 1$ and leak \((setup_{2}\text{-}B, sid)\) to $Sim$.
        \item Upon receiving input $(publish\text{-}B, sid)$ from $B$, if $setup_{B}^{r4S} = 1 \wedge published_{B}^{r4S} = 0$, set $published_{B}^{r4S} \leftarrow 1$ and leak $(publish\text{-}B, sid)$ to $Sim$.
        \item Upon receiving input $(publish\text{-}A, sid)$ from $B$, if $setup^{r4S} = 1 \wedge published_{A}^{r4S} = 0$, set $published_{A}^{r4S} \leftarrow 1$ and leak $(publish\text{-}A, sid)$ to $Sim$.
        \item Upon receiving input $(init\text{-}B, sid)$ from $M_1$, if $published_{B}^{r4S} = 1 \wedge init_{B}^{r4S} = 0$, set $init_{B}^{r4S} \leftarrow 1$ and leak $(init\text{-}B, sid)$ to $Sim$.
        \item Upon receiving input $(init\text{-}A, sid)$ from $M_2$, if $published_{A}^{r4S} = 1 \wedge init_{A}^{r4S} = 0$, set $init_{A}^{r4S} \leftarrow 1$ and leak $(init\text{-}A, sid)$ to $Sim$.
        \item Upon receiving input $(share, sid)$ from $A$, if $init_{A}^{r4S} = 1 \wedge shared^{r4S} = 0$, set $shared^{r4S} \leftarrow 1$ and leak $(share, sid)$ to $Sim$.
        \item Upon receiving input $(redeem, sid, path)$ from party $P$ where $path \in paths$:
            \\When $init_B^{r4S} = 1$:
            \begin{itemize}
                \item Update $pub_m^{r4S} \leftarrow pub_m^{r4S} \vee ((P=B) \wedge shared^{r4S} \wedge (path = eclaim_A\text{-}B)) \vee ((P=B) \wedge (path = claim_A\text{-}B))$.
                \item Update $pub_e^{r4S} \leftarrow pub_e^{r4S} \vee ((P=B) \wedge shared^{r4S} \wedge (path = eclaim_A\text{-}B))$.
                \item Update $pub_{r1}^{r4S} \leftarrow pub_{r1}^{r4S} \vee ((P=A) \wedge path \in \{erefund_A\text{-}A, refund_A\text{-}A\})$.
            \end{itemize}
        When $init_B^{r4S} = 1$ and $init_A^{r4S} = 1$:
        \begin{itemize}
            \item Update $pub_{br}^{r4S} \leftarrow pub_{br}^{r4S} \vee ((P=A) \wedge (path = claim_B\text{-}A))$.
            \item Update $pub_{r2}^{r4S} \leftarrow pub_{r2}^{r4S} \vee ((P=B) \wedge path = refund_B\text{-}B)$.
        \end{itemize}
        Set $res^{r4S}$ $\leftarrow$ $rPredicate$$($$path$, $P$, $pub_{m}^{r4S}$, $pub_{r1}^{r4S}$, $pub_{r2}^{r4S}$, $pub_{e}^{r4S}$, $pub_{br}^{r4S}$$)$, leak $(redeem$, $sid$, $path)$ to $Sim$ and send $res^{r4S}$ to $P$.
        \item Upon receiving $(update, sid, pre)$ from $Sim$ where $pre \in \{m, r1, r2, e, br\}$ via the influence port, set $pub_{pre}^{r4S} \leftarrow 1$.
    \end{itemize}
    \end{minipage}
}
\caption{Ideal Functionality $\mathcal{F}_{r4S}$}
\label{fig:idealfunctionality4S1}
\end{figure}

\subsection{Security Proof}
This section outlines the proof of Lemma~\ref{mainlemma} and is organized as follows. First, we establish the security in the case of static corruption. We then explain how protocols and ideal functionalities follow a phase-based approach, where they expect specific messages in each phase. Honest parties and ideal functionalities disregard messages that are unexpected or received out of sequence. The section begins with a description of the global ideal functionalities \(\mathcal{H}\) and \(\mathcal{G}_{mbp}\). We then formalize the protocol \(\Pi_{r4S}\), define \(\mathcal{F}_{r4S}\), and prove that it satisfies the security guarantees described in Lemma 1. Lastly, we show that \(\Pi_{r4S}\) UC-realizes \(\mathcal{F}_{r4S}\) by providing construction of simulator.

\subsubsection{Functionalities}  

    \paragraph{Global Random Oracle}
    In the security analysis, we model \(\mathcal{H}\) as a global random oracle, described as an ideal functionality \(\mathcal{H}\)~\cite{canetti2014practical}. This functionality accepts queries of different lengths and returns outputs of fixed length \(\mu\). The response is chosen uniformly at random for each new query, and repeated queries for the same input always return the same result. Figure~\ref{fig:globalrandomoracle} outlines the details of \(\mathcal{H}\).
    \begin{figure}
    \centering
    \fbox{
    \begin{minipage}{\columnwidth}
    The functionality \( \mathcal{H} \) is a global random oracle that takes inputs \( x \in \{0, 1\}^* \), producing outputs \( y \in \{0, 1\}^\mu \), and maintains an initially empty table \( T \) to record query-response mappings.
    \begin{itemize}
        \item Upon receiving \( (sid, x) \) from a party or functionality:
        \begin{itemize}
            \item If \( x \) has been queried before and \( (x, y) \in T \), return the stored response \( y \).
            \item If \( x \) is a new query, sample \( y \in \{0,1\}^\mu \) uniformly at random, store \( (x, y) \) in \( T \), and return \( y \).
        \end{itemize}
    \end{itemize}
    \end{minipage}
    }
    \caption{Global Random Oracle $\mathcal{H}$}
    \label{fig:globalrandomoracle}
    \end{figure}

    \paragraph{Mempool and Blockchain Projection $\mathcal{G}_{mbp}$}

    The functionality \(\mathcal{G}_{mbp}\) models the interaction between two parties, \(A\) and \(B\), and their respective blockchains, focusing on simulating mempool projections during cross-chain activities. It manages the creation, publication, and redemption of locking transactions while providing adversarial visibility into the system's state transitions.

    During the setup phase, \(B\) initiates the process by sending a setup message containing transaction hashes (\(dig_m\) and \(dig_{r2}\)) to \(\mathcal{G}_{mbp}\). These hashes are recorded, leaked to the adversary (\(Adv\)), and forwarded to \(A\). Subsequently, \(A\) completes its setup by sending additional transaction hashes (\(dig_e\), \(dig_{r1}\), and \(dig_{br}\)), which the functionality uses to create a locking transaction (\(tx_A^{lock}\)) and a corresponding claim transaction (\(tx_A^{claim}\)). These transactions are sent to \(B\) and leaked to \(Adv\). In a second setup phase, \(B\) gets its locking and claim transactions (\(tx_B^{lock}\) and \(tx_B^{claim}\)) based on the recorded digests, the functionality also sends them to \(A\), and triggers another leak to \(Adv\).
    
    The publication and initialization stages allow \(A\) and \(B\) to publish their transactions, making them visible to all parties. Upon receiving a publication request, \(\mathcal{G}_{mbp}\) updates the publication state and leaks the event to \(Adv\). Miners (\(M_1\) and \(M_2\)) can then confirm these transactions in the respective blockchains, further propagating the information to all parties while also informing \(Adv\) about the event.
    
    In the redemption phase, any party, including \(A\), \(B\), or miners, can attempt to redeem a locking transaction by providing secrets and a valid path. \(\mathcal{G}_{mbp}\) verifies the secrets against their corresponding digests using a hash function \(\mathcal{H}\). It updates the publication states based on the validity of the secrets, computes the redemption result using a predefined predicate (\(rPredicate\)), and leaks all relevant details to \(Adv\). The computed result is then sent back to the requesting party.

\begin{figure*}
\centering
\fbox{ 
    \begin{minipage}{\textwidth}
    The ideal functionality \( \mathcal{G}_{mbp} \) in the hybrid model simulates the mempool projection of two contracts across different blockchains for a session ID \( sid \). It interacts with parties \( A, B \), miners \( M_1 \) and \( M_2 \) in the respective blockchains, the functionality \( \mathcal{H} \), and the adversary \( Adv \).
    It has the following variables $dig_m^{mbp} \leftarrow \perp$ and $dig_{r2}^{mbp} \leftarrow \perp$, $dig_e^{mbp} \leftarrow \perp$, $dig_{r1}^{mbp} \leftarrow \perp$, $dig_{br}^{mbp} \leftarrow \perp$, $tx_A^{lock} \leftarrow \perp$, $tx_B^{lock} \leftarrow \perp$, $publish_B^{mbp} \leftarrow 0$, $publish_A^{mbp} \leftarrow 0$, $init_B^{mbp} \leftarrow 0$, $init_A^{mbp} \leftarrow 0$, $pub_m^{mbp} \leftarrow 0$,$pub_e^{mbp} \leftarrow 0$,$pub_{r1}^{mbp} \leftarrow 0$,$pub_{br}^{mbp} \leftarrow 0$,$pub_{r2}^{mbp} \leftarrow 0$ and $res^{mbp} \leftarrow 0$.
    \begin{itemize}
        \item Upon receiving $(setup\text{-}B, sid, dig_m, dig_{r2})$ from $B$ when $dig_m^{mbp} = \perp$ and $dig_{r2}^{mbp} = \perp$, set $dig_m^{mbp} \leftarrow dig_m$, $dig_{r2}^{mbp} \leftarrow dig_{r2}$. Leak $(setup\text{-}B, sid, dig_m, dig_{r2})$ to $Adv$ and send $(setup\text{-}B, sid, dig_m, dig_{r2})$ to $A$.
        \item Upon receiving $(setup\text{-}A, sid, dig_{e}, dig_{r1}, dig_{br})$ from $A$ when $dig_e^{mbp} = \perp$, $dig_{r1}^{mbp} = \perp$ and $dig_{br}^{mbp} = \perp$, set $dig_e^{mbp} \leftarrow dig_e$, $dig_{r1}^{mbp} \leftarrow dig_{r1}$, $dig_{br}^{mbp} \leftarrow dig_{br}$. Create $tx_A^{lock}$ with parameters $dig_m^{mbp}, dig_e^{mbp}, dig_{r1}^{mbp}, dig_{r2}^{mbp}, dig_{br}^{mbp}$ and store it. Then create $tx_A^{claim}$ using $tx_A^{lock}$ and send both transactions to $B$. Leak $(setup\text{-}A, sid, dig_{e}, dig_{r1}, dig_{br})$ to $Adv$.
        \item Upon receiving $(setup_2\text{-}B, sid)$ from $B$ when $dig_e^{mbp} \neq \perp$, $dig_{r1}^{mbp} \neq \perp$ and $dig_{br}^{mbp} \neq \perp$, create $tx_B^{lock}$ with parameters $dig_m^{mbp}, dig_{r1}^{mbp}, dig_{r2}^{mbp}, dig_{br}^{mbp}$ and store it. Then create $tx_B^{claim}$ using $tx_B^{lock}$ and send both to $A$. Leak $(setup_2\text{-}B, sid)$ to $Adv$.
        \item Upon receiving \( (publish\text{-}B, sid, tx_A^{lock}) \) from \( B \), if \( publish_B^{mbp} = 0 \), set \( publish_B^{mbp} \leftarrow 1 \), leak \( (publish\text{-}B, sid) \) to $Adv$, and send \( (publish\text{-}B, sid, tx_A^{lock}) \) to all parties.
        \item Upon receiving $(init\text{-}B, sid, tx_A^{lock})$ from $M_1$ when $publish_B^{mbp}=1$, set $init_B^{mbp} \leftarrow 1$, leak $(init\text{-}B, sid, tx_A^{lock})$ to $Adv$ and send $(init\text{-}B, sid, tx_A^{lock, mbp})$ to all parties.
        \item Upon receiving $(publish\text{-}A, sid, tx_B^{lock})$ from $B$. If $publish_A^{mbp}=0$, set $publish_A^{mbp} \leftarrow 1$, leak $(publish\text{-}A, sid)$ to $Adv$ and send $(publish\text{-}A, sid, tx_B^{lock})$ to all parties.
        \item Upon receiving $(init\text{-}A, sid, tx_B^{lock})$ from $M_2$ when $publish_A^{mbp}=1$, set $init_A^{mbp} \leftarrow 1$, leak $(init\text{-}A, sid, tx_B^{lock, mbp})$ to $Adv$ and send $(init\text{-}A, sid, tx_B^{lock})$ to all parties. 
        \item Upon receiving $(redeem$, $sid$, $tx_A^{lock}$, $s_m$, $s_e$, $s_{r1}$, $s_{br}$, $s_{r2}$, $path)$ or $(redeem$, $sid$, $tx_B^{lock}$, $s_m$, $s_{r1}$, $s_{br}$, $s_{r2}$, $path)$ from any party $P$ $\in$ $\{$$A$, $B$, $M_1$, $M_2$$\}$ such that $path \in paths$:
        \begin{itemize}
            \item $pub_m^{mbp} \leftarrow pub_m^{mbp} \vee ((s_m \neq \perp) \wedge (\mathcal{H}(sid, s_m) = dig_m^{mbp}))$
            \item $pub_e^{mbp} \leftarrow pub_e^{mbp} \vee ((s_e \neq \perp) \wedge (\mathcal{H}(sid, s_e) = dig_e^{mbp}))$
            \item $pub_{r1}^{mbp} \leftarrow pub_{r1}^{mbp} \vee ((s_{r1} \neq \perp) \wedge (\mathcal{H}(sid, s_{r1}) = dig_{r1}^{mbp}))$
            \item $pub_{br}^{mbp} \leftarrow pub_{br}^{mbp} \vee ((s_{br} \neq \perp) \wedge (\mathcal{H}(sid, s_{br}) = dig_{br}^{mbp}))$
            \item $pub_{r2}^{mbp} \leftarrow pub_{r2}^{mbp} \vee ((s_{r2} \neq \perp) \wedge (\mathcal{H}(sid, s_{r2}) = dig_{r2}^{mbp}))$.
            \item $res^{mbp}$ $\leftarrow$ $rPredicate$$($$path$, $P$, $s_m$, $s_e$, $s_{r1}$, $s_{br}$, $s_{r2})$, leak $(redeem$, $sid$, $s_m$, $s_e$, $s_{r1}$, $s_{br}$, $s_{r2}$, $path$,$P$, $pub_m^{mbp}$, $pub_e^{mbp}$, $pub_{r1}^{mbp}$, $pub_{br}^{mbp}$, $pub_{r2}^{mbp}$, $res^{mbp})$ to $Adv$ and send $res^{mbp}$ to $P$.
        \end{itemize} 
    \end{itemize}
    \end{minipage}
}
\caption{Mempool and blockchain projection functionality $\mathcal{G}_{mbp}$}
\label{fig:gmbp}
\end{figure*}

    \paragraph{Relaxed 4-Swap $\Pi_{r4S}$}

    The protocol \(\Pi_{r4S}\) describes the 4-Swap protocol between two participants, \(A\) and \(B\), with the involvement of two miners, \(M_1\) and \(M_2\) under relaxed conditions as discussed in Section \ref{ucframework}. The process is divided into three phases: setup, initiation, and redemption. The protocol ensures secure interaction between participants using ideal functionalities \(\mathcal{H}\) for hashing and \(\mathcal{G}_{mbp}\) for managing mempool projections and cross-chain transactions. Each participant and miner maintains a local state, which transitions through predefined stages: setup, initiated, and redeeming.

    In the setup phase, \(A\) and \(B\) initiate the protocol by generating and sharing cryptographic digests of random secrets via \(\mathcal{H}\). These digests are sent to \(\mathcal{G}_{mbp}\), which uses them to create locking and claim transactions. \(B\) begins the process by creating digests (\(dig_m, dig_{r2}\)) from its secrets and sending them to \(\mathcal{G}_{mbp}\). In response, \(A\) stores these digests and generates its own (\(dig_e, dig_{r1}, dig_{br}\)), which are similarly processed by \(\mathcal{G}_{mbp}\). Both parties store the locking and claim transactions provided by \(\mathcal{G}_{mbp}\), transitioning their states to initiated.
    
    In the initiation phase, participants and miners prepare to initiate the swap. \(B\) publishes its locking transaction (\(tx_A^{lock}\)) through \(\mathcal{G}_{mbp}\), prompting miners (\(M_1\) and \(M_2\)) to acknowledge and store these transactions. Once \(M_1\) confirms receipt of \(tx_A^{lock}\), the state of $B$ is set to redeeming. Similarly, \(A\) waits for \(B\)'s locking transaction to be published and confirmed before publishing its own (\(tx_B^{lock}\)). \(M_2\) takes analogous steps, ensuring both locking transactions are available for redemption.
    
    In the redemption phase, $A$ may share secret $s_e$ with $B$. Depending on the chosen redemption path, participants and miners submit secrets to \(\mathcal{G}_{mbp}\) for transaction redemption. These secrets are validated against their corresponding digests, and the outcomes are broadcast. Both locking transactions are redeemed, completing the atomic swap.

\begin{figure}
\centering
\fbox{ 
    \begin{minipage}{\columnwidth}
    Protocol $\Pi_{r4S}$ defines an atomic swap between participants $A$, $B$, $M_1$, and $M_2$ for session $sid$, structured in three phases with two miners (possibly instances of the same one). Participants interact with ideal functionalities \(\mathcal{H}\) and \(\mathcal{G}_{mbp}\), and each have a local state $state^{r4S}$. Initially, $A$ and $B$ start with $state_A^{r4S} = state_B^{r4S} =$ setup, while $M_1$ and $M_2$ begin with $state_{M_1}^{r4S} = state_{M_2}^{r4S} =$ initiated.
    
    \rule{\textwidth}{0.4pt}
    \begin{center} $\boldsymbol{state_{P}^{r4S}}$ \textbf{= setup where } $\boldsymbol{P \in \{A, B\}}$ \end{center}
    $\boldsymbol{B}$: 
    \begin{itemize}
        \item Upon receiving $(setup\text{-}B, sid)$ from $\mathcal{Z}$, sample $s_m \stackrel{\text{R}}{\leftarrow} \{0, 1\}^{\mu}$ and $s_{r2} \stackrel{\text{R}}{\leftarrow} \{0, 1\}^{\mu}$. Compute $dig_m \leftarrow \mathcal{H}(sid, s_m)$ and $dig_{r2} \leftarrow \mathcal{H}(sid, s_{r2})$, then transmit $(setup\text{-}B, sid, dig_m, dig_{r2})$ to $\mathcal{G}_{mbp}$. Finally, await $tx_A^{claim}$ and $tx_A^{claim}$ from $\mathcal{G}_{mbp}$. Then store $tx_A^{lock}$ and $tx_A^{claim}$ and set $state_B^{r4S} \leftarrow$ initiated.
        \item Upon receiving $(setup\text{-}A, sid, dig_e, dig_{r1}, dig_{br})$ from $\mathcal{G}_{mbp}$, store $dig_e, dig_{r1}, dig_{br}$.
    \end{itemize}

    $\boldsymbol{A}$: 
    \begin{itemize}
        \item Upon receiving $(setup\text{-}B, sid, dig_m, dig_{r2})$ from $\mathcal{G}_{mbp}$, store $dig_m, dig_{r2}$.
        \item Upon receiving $(setup\text{-}A, sid)$ from $\mathcal{Z}$, sample $s_e \stackrel{\text{R}}{\leftarrow} \{0, 1\}^{\mu}$, $s_{r1} \stackrel{\text{R}}{\leftarrow} \{0, 1\}^{\mu}$ and $s_{br} \stackrel{\text{R}}{\leftarrow} \{0, 1\}^{\mu}$. Compute $dig_{e} \leftarrow \mathcal{H}(sid, s_e)$, $dig_{r1} \leftarrow \mathcal{H}(sid, s_{r1})$ and $dig_{br} \leftarrow \mathcal{H}(sid, s_{br})$, then transmit $(setup\text{-}B, sid, dig_e, dig_{r1}, dig_{br})$ to $\mathcal{G}_{mbp}$. Finally, await $tx_B^{claim}$ and $tx_B^{claim}$ from $\mathcal{G}_{mbp}$. Then store $tx_B^{lock}$ and $tx_B^{claim}$ and set $state_A^{r4S} \leftarrow$ initiated.
        \item Upon receiving $(init\text{-}B, sid, tx_A^{lock})$ from $\mathcal{G}_{mbp}$, set $state_A^{r4S} \leftarrow$ redeeming.
    \end{itemize}

    \rule{\textwidth}{0.4pt}
    \begin{center} $\boldsymbol{state_{P}^{r4S}}$ \textbf{= initiated where } $\boldsymbol{P \in \{A, B, M_1, M_2\}}$ \end{center}
    $\boldsymbol{B}$: 
    \begin{itemize}
        \item Upon receiving $(setup_2\text{-}B, sid)$ from $\mathcal{Z}$, after receiving $(setup\text{-}A, sid, dig_e, dig_{r1}, dig_{br})$, send $(setup_2\text{-}B, sid)$ to $\mathcal{G}_{mbp}$.
        \item Upon receiving $(publish\text{-}B, sid)$ from $\mathcal{Z}$, send $(publish\text{-}B, sid, tx_A^{lock})$ to $\mathcal{G}_{mbp}$.
        \item Upon receiving $(init\text{-}B, sid, tx_A^{lock})$ from $\mathcal{G}_{mbp}$, $state_B^{r4S} \leftarrow$ redeeming.
    \end{itemize}
    $\boldsymbol{A}$: 
    \begin{itemize}
        \item Upon receiving $(init\text{-}B, sid, tx_A^{lock})$ from $\mathcal{G}_{mbp}$, set $published_B^{r4S} \leftarrow 1$.
        \item Upon receiving $(publish\text{-}A, sid)$ from $\mathcal{Z}$ and $published_B^{r4S} = 1$, send $(publish\text{-}A, sid, tx_B^{lock})$ to $\mathcal{G}_{mbp}$.
        \item Upon receiving $(init\text{-}A, sid, tx_B^{lock})$ from $\mathcal{G}_{mbp}$, $state_A^{r4S} \leftarrow$ redeeming.
    \end{itemize}
    \end{minipage}
}
\caption{Protocol $\Pi_{r4S}$}
\label{fig:relaxedformalprotocol4S}
\end{figure}

\begin{figure}
\centering
\fbox{ 
    \begin{minipage}{\columnwidth}
    $\boldsymbol{M_1}$: 
    \begin{itemize}
        \item Upon receiving $(publish\text{-}B, sid, tx_A^{lock})$ from $\mathcal{G}_{mbp}$, store $tx_A^{lock}$ and set $received_B^{r4S} \leftarrow 1$.
        \item Upon receiving $(init\text{-}B, sid)$ from $\mathcal{Z}$ and $received_B^{r4S} = 1$, send $(init\text{-}B, sid, tx_A^{lock})$ to $\mathcal{G}_{mbp}$ and set $state_{M_1}^{r4S} \leftarrow$ redeeming.
    \end{itemize}
    $\boldsymbol{M_2}$: 
    \begin{itemize}
        \item Upon receiving $(publish\text{-}A, sid, tx_B^{lock})$ from $\mathcal{G}_{mbp}$, store $tx_B^{lock}$ and set $received_A^{r4S} \leftarrow 1$.
        \item Upon receiving $(init\text{-}A, sid)$ from $\mathcal{Z}$ and $received_A^{r4S} = 1$, send $(init\text{-}A, sid, tx_B^{lock})$ to $\mathcal{G}_{mbp}$ and set $state_{M_2}^{r4S} \leftarrow$ redeeming.
    \end{itemize}
    
    \rule{\textwidth}{0.4pt} 
    \begin{center} $\boldsymbol{state_{P}^{r4S}}$ \textbf{= redeeming where } $\boldsymbol{P \in \{A, B, M_1, M_2\}}$ \end{center}
    $\boldsymbol{A}$: \begin{itemize}
        \item Upon receiving $(share, sid)$ from $\mathcal{Z}$, send $(share, sid, s)$ to $B$.
    \end{itemize}
    $\boldsymbol{B}$: 
    \begin{itemize}
        \item Upon receiving $(share, sid, s)$ from $A$, verify $dig_{e} = \mathcal{H}(sid, s)$, set $s_e \leftarrow s$.
    \end{itemize}
    $\boldsymbol{P \in \{A, B, M_1, M_2\}}$: 
    \begin{itemize}
        \item Upon receiving $(redeem, sid, path)$ from $\mathcal{Z}$ such that $path \in \mathcal{Q}$:
        \begin{itemize}
            \item If \( path \in \{claim_B\text{-}A, claim_A\text{-}B, eclaim_A\text{-}B, slash_1^{r4S}\), \(slash_2^{r4s}\} \), with \( slash_1^{r4S} \in slash_1 \) and \( slash_2^{r4S} \in slash_2 \), and \( s_m \) is already available, then \( p_m^{r4S} \leftarrow s_m \); otherwise, it is set to \( \perp \).
            \item If $path \in  \{eclaim_A\text{-}B\}$, and $s_e$ is already available, then $p_e^{r4S} \leftarrow s_e$; otherwise, it is set to \( \perp \).
            \item If \( path \in \{refund_A\text{-}A, erefund_A\text{-}A, slash_1^{r4S}\} \), with \( slash_1^{r4S} \in slash_1 \), and \( s_{r1} \) is already available, then \( p_{r1}^{r4S} \leftarrow s_{r1} \); otherwise, it is set to \( \perp \).
            \item If $path \in  \{claim_B\text{-}A, slash_2^{r4s}\}$, with \( slash_2^{r4S} \in slash_2 \), and $s_{br}$ is already available, then $p_{br}^{r4S} \leftarrow s_{br}$; otherwise, it is set to \( \perp \).
            \item If $path \in  \{refund_B\text{-}B, slash_2^{r4s}\}$, with \( slash_2^{r4S} \in slash_2 \), and $s_{r2}$ is already available, then $p_{r2}^{r4S} \leftarrow s_{r2}$; otherwise, it is set to \( \perp \).
            \item Send \( (redeem, sid, tx_A^{lock}, p_m^{r4S}, p_e^{r4S}, p_{r1}^{r4S}, p_{br}^{r4S}, p_{r2}^{r4S}\), \(path) \) to \( \mathcal{G}_{mbp} \) and display the result.
            \item Send $(redeem, sid, tx_B^{lock}, p_m^{r4S}, p_{r1}^{r4S}, p_{br}^{r4S}, p_{r2}^{r4S}, path)$ to $\mathcal{G}_{mbp}$, return result.
        \end{itemize}
    \end{itemize}
    \end{minipage}
}
\caption{Protocol $\Pi_{r4S}$ cont.}
\label{fig:relaxedformalprotocol4S1}
\end{figure}
    
    \paragraph{Indistinguishability Proof}
    \begin{lemma}
        Let \( \mathcal{Z} \) be any probabilistic polynomial-time (PPT) environment and \( Adv \) be any PPT adversary that corrupts any subset of the parties \( \{A, B, M_1, M_2\} \). Then there exists a PPT simulator \( Sim \) such that the execution of the protocol \( \Pi_{r4S} \) in the hybrid world with \( Adv \) is computationally indistinguishable from the execution of the ideal functionality \( \mathcal{F}_{r4S} \) in the ideal world with \( Sim \).
    \end{lemma}
    To validate our protocol, we construct a simulator, $Sim$, that effectively bridges the gap between the hybrid and ideal worlds. In the hybrid world, both honest and corrupted parties interact not only with each other but also with the ideal functionalities \( H \) and \( \mathcal{G}_{mbp} \). Conversely, in the ideal world, honest parties communicate solely with the functionality \( \mathcal{F}_{r4S} \).

$Sim$ has the critical role of replicating the hybrid world's interactions within the ideal world to ensure that the corrupted parties cannot distinguish between the two settings. This involves two main tasks:

\begin{itemize}
    \item Simulating Honest Parties for Corrupted Parties: $Sim$ observes the actions of the honest parties through the leakage provided by \( \mathcal{F}_{r4S} \). Using this information, it sends messages to the corrupted parties that mimic those the honest parties would send if they were executing the protocol \( \Pi_{\text{r4S}} \). This ensures that the corrupted parties receive consistent and indistinguishable communications from what they expect in the hybrid world.
    \item Handling Corrupted Parties' Messages to Honest Parties: When corrupted parties send messages intended for the ideal functionalities or the honest parties, $Sim$ intercepts these messages. It then simulates their impact by providing appropriate inputs to \( \mathcal{F}_{r4S} \) on behalf of the corrupted parties or by utilizing the influence port of \( \mathcal{F}_{r4S} \). This accurately reflects the corrupted parties' actions within the ideal world.
\end{itemize}

By performing these tasks, $Sim$ ensures that the interactions experienced by the corrupted parties in the ideal world are indistinguishable from those in the hybrid world. This demonstrates that our protocol securely maintains the desired properties even in the presence of corrupted parties.

\begin{proof}
To ensure that executions in the ideal world are indistinguishable from those in the hybrid world, we design a simulator \( \text{Sim} \) capable of handling any number of corrupted parties. \( Sim\) internally simulates necessary components like \( H \) and \( \mathcal{G}_{mbp} \) and reacts to events in a way that mirrors the hybrid world.

In the hybrid world, the environment $\mathcal{Z}$ learns about the existence of messages exchanged between parties through the adversary \( Adv \), but not their contents. Honest parties send messages based on inputs from $\mathcal{Z}$ according to the protocol \( \Pi_{r4S} \), while corrupted parties are controlled by \( Adv \). Messages sent to corrupted parties are intercepted by \( Adv \), who may reveal them to $\mathcal{Z}$.

In contrast, honest parties in the ideal world do not communicate with each other; they only send their inputs to the ideal functionality \( F_{r4S} \). To replicate the communication patterns observed in the hybrid world, \( Sim \) utilises the inputs that honest parties send to \( F_{r4S} \) through leakage. This enables \( Sim \) to mimic the same communication as in the hybrid world.

The remaining task is demonstrating that the messages generated by \( \text{Sim} \) are indistinguishable from those in the hybrid world.
    \paragraph{State indistinguisability}
    The parties, irrespective of honest or corrupted, interact with the functionalities $\mathcal{H}$ and $\mathcal{G}_{mbp}$, but in the ideal world, only the corrupted parties interact with the internally simulated functionalities. In contrast, the honest parties only interact with the $\mathcal{F}_{r4S}$. We only discuss the different cases of simulator construction in the ideal world to make it indistinguishable from the hybrid world.
    \begin{itemize}
        \item $B\text{-}setup$:In the hybrid world $B$ sends $(setup\text{-}B, sid, dig_m, dig_{r2})$ to $\mathcal{G}_{mbp}$ irrespective of being honest or corrupted. In the ideal world:
            \begin{itemize}
                \item When $B$ is honest, $\mathcal{F}_{r4S}$ leaks $(setup\text{-}B, sid)$ to $Sim$ which internally simulates $B$ and samples $s_m^{sim} \stackrel{\text{R}}{\leftarrow} \{0, 1\}^{\mu}$, $s_{r2}^{sim} \stackrel{\text{R}}{\leftarrow} \{0, 1\}^{\mu}$ and sets $dig_{m}^{sim} \leftarrow \mathcal{H}(sid, s_m^{sim})$, $dig_{r2}^{sim} \leftarrow \mathcal{H}(sid, s_{r2}^{sim})$ through internally simulated $\mathcal{H}$. Simulate sending $(setup\text{-}B, sid, dig_m^{sim}, dig_{r2}^{sim})$ to $\mathcal{G}_{mbp}$ and leaking to $Adv$.
                \item When $B$ is corrupted, $Sim$ sets $dig_m^{sim} \leftarrow dig_m$, $dig_{r2}^{sim}$ $\leftarrow$ $dig_{r2}$ after receiving $(setup\text{-}B, sid, dig_m, dig_{r2})$ sent to $\mathcal{G}_{mbp}$ and send $(setup\text{-}B, sid)$ as honest $B$ to $\mathcal{F}_{r4S}$.
                \item For both cases, set $setup_B^{r4S} \leftarrow 1$.
            \end{itemize}
        \item $A\text{-}setup$: In the hybrid world $A$ sends $(setup\text{-}A$, $sid$, $dig_e$, $dig_{r1}$, $dig_{br})$ to $\mathcal{G}_{mbp}$ irrespective of being honest or corrupted. In the ideal world: 
            \begin{itemize}
                \item When $A$ is honest, $\mathcal{F}_{r4S}$ leaks $(setup\text{-}A, sid)$ to $Sim$ which internally simulates $A$ and samples $s_e^{sim}$ $\stackrel{\text{R}}{\leftarrow} \{0, 1\}^{\mu}$, $s_{r1}^{sim}$ $\stackrel{\text{R}}{\leftarrow} \{0, 1\}^{\mu}$, $s_{br}^{sim}$ $\stackrel{\text{R}}{\leftarrow} \{0, 1\}^{\mu}$ and sets $dig_{e}^{sim}$ $\leftarrow$ $\mathcal{H}$$(sid$, $s_e^{sim})$, $dig_{r1}^{sim}$ $\leftarrow$ $\mathcal{H}$$(sid$, $s_{r1}^{sim})$, $dig_{br}^{sim}$ $\leftarrow$ $\mathcal{H}$$(sid$, $s_{br}^{sim})$ through internally simulated $\mathcal{H}$. Simulate sending $(setup\text{-}A$, $sid$, $dig_e^{sim}$, $dig_{r1}^{sim}$, $dig_{br}^{sim})$ to $\mathcal{G}_{mbp}$ and leaking to $Adv$.
                \item When $A$ is corrupted, $Sim$ sets $dig_e^{sim}$ $\leftarrow$ $dig_e$, $dig_{r1}^{sim}$ $\leftarrow$ $dig_{r1}$, $dig_{br}^{sim}$ $\leftarrow$ $dig_{br}$ after receiving $(setup\text{-}A$, $sid$, $dig_e$, $dig_{r1}$, $dig_{br})$ sent to $\mathcal{G}_{mbp}$ and send $(setup\text{-}A, sid)$ as honest $A$ to $\mathcal{F}_{r4S}$.
                \item $tx_{A}^{lock, sim}$, $tx_{A}^{claim, sim}$ is set by internally simulating $\mathcal{G}_{mbp}$ on $dig_m^{sim}$, $dig_{r2}^{sim}$, $dig_e^{sim}$, $dig_{r1}^{sim}$, $dig_{br}^{sim}$.
                \item For both cases, set $setup_A^{r4S} \leftarrow 1$.
            \end{itemize}
        \item $B\text{-}setup_2$: In the hybrid world $B$ sends $(setup_2\text{-}B, sid)$ to $\mathcal{G}_{mbp}$ irrespective of being honest or corrupted. In the ideal world: 
            \begin{itemize}
                \item When $B$ is honest, $\mathcal{F}_{r4S}$ leaks $(setup_2\text{-}B, sid)$ to $Sim$ which internally simulates $B$ sending $(setup_2\text{-}B, sid)$ to simulated $\mathcal{G}_{mbp}$.
                \item When $B$ is corrupted and $Sim$ receives $(setup_2\text{-}B, sid)$ sent to $\mathcal{G}_{mbp}$, applies the message on internally simulated $\mathcal{G}_{mbp}$  and passes $(setup_2\text{-}B, sid)$ to  $\mathcal{F}_{r4S}$ as honest $B$.
                \item $tx_{B}^{lock, sim}$, $tx_{B}^{claim, sim}$ is set by internally simulating $\mathcal{G}_{mbp}$ on $dig_m^{sim}$, $dig_{r2}^{sim}$, $dig_{r1}^{sim}$, $dig_{br}^{sim}$.
                \item For both cases, set $setup^{r4S} \leftarrow 1$.
            \end{itemize}
        \item $B\text{-}publish$: In the hybrid world $B$ sends $(publish\text{-}B, sid, tx_A^{lock})$ to $\mathcal{G}_{mbp}$ irrespective of being honest or corrupted. In the ideal world: 
            \begin{itemize}
                \item When $B$ is honest, $\mathcal{F}_{r4S}$ leaks $(publish\text{-}B, sid)$ to $Sim$ which internally simulates $B$ sending $(publish\text{-}B, sid)$ to simulated $\mathcal{G}_{mbp}$.
                \item When $B$ is corrupted, $Sim$ internally simulates $\mathcal{G}_{mbp}$, verifies $tx_A^{lock} = tx_A^{lock, sim}$ and then sends $(publish\text{-}B, sid)$ as honest $B$ to $\mathcal{F}_{r4S}$.
                \item For both, set $published_B^{r4S} \leftarrow 1$ and $publish_B^{mbp} \leftarrow 1$.
            \end{itemize}
        \item $A\text{-}publish$: In the hybrid world $A$ sends $(publish\text{-}A, sid, tx_B^{lock})$ to $\mathcal{G}_{mbp}$ irrespective of being honest or corrupted. In the ideal world: 
            \begin{itemize}
                \item When $A$ is honest, $\mathcal{F}_{r4S}$ leaks $(publish\text{-}A, sid)$ to $Sim$ which internally simulates $A$ sending $(publish\text{-}A, sid)$ to simulated $\mathcal{G}_{mbp}$.
                \item When $A$ is corrupted, $Sim$ internally simulates $\mathcal{G}_{mbp}$, verifies $tx_B^{lock} = tx_B^{lock, sim}$ and then sends $(publish\text{-}A, sid)$ as honest $A$ to $\mathcal{F}_{r4S}$.
                \item For both, set $published_A^{r4S} \leftarrow 1$ and $publish_A^{mbp} \leftarrow 1$.
            \end{itemize}
        \item $M_1\text{-}init$: In the hybrid world $M_1$ sends $(init\text{-}B, sid, tx_A^{lock})$ to $\mathcal{G}_{mbp}$ irrespective of being honest or corrupted. In the ideal world: 
            \begin{itemize}
                \item When $M_1$ is honest, $\mathcal{F}_{r4S}$ leaks $(init\text{-}B, sid)$ to $Sim$ which internally simulates $M_1$ sending $(init\text{-}B, sid, tx_A^{lock})$ to simulated $\mathcal{G}_{mbp}$.
                \item When $M_1$ is corrupted, $Sim$ internally simulates $\mathcal{G}_{mbp}$, verifies $tx_A^{lock} = tx_A^{lock, sim}$ and then sends $(init\text{-}B, sid)$ as honest $M_1$ to $\mathcal{F}_{r4S}$.
                \item For both cases, set $init_B^{r4S} \leftarrow 1$ and $init_B^{mbp} \leftarrow 1$.
            \end{itemize}
        \item $M_2\text{-}init$: In the hybrid world $M_2$ sends $(init\text{-}A, sid, tx_B^{lock})$ to $\mathcal{G}_{mbp}$ irrespective of being honest or corrupted. In the ideal world: 
            \begin{itemize}
                \item When $M_2$ is honest, $\mathcal{F}_{r4S}$ leaks $(init\text{-}A, sid)$ to $Sim$ which internally simulates $M_2$ sending $(init\text{-}A, sid, tx_B^{lock})$ to simulated $\mathcal{G}_{mbp}$.
                \item When $M_2$ is corrupted, $Sim$ internally simulates $\mathcal{G}_{mbp}$, verifies $tx_B^{lock} = tx_B^{lock, sim}$ and then sends $(init\text{-}A, sid)$ as honest $M_2$ to $\mathcal{F}_{r4S}$.
                \item For both cases, set $init_A^{r4S} \leftarrow 1$ and $init_A^{mbp} \leftarrow 1$.
            \end{itemize}
        \item $A\text{-}share$: In the hybrid world $A$ sends $(share, sid, s)$ to $B$ irrespective of being honest or corrupted. In the ideal world: 
        \begin{itemize}
            \item When $A$ is honest, $\mathcal{F}_{r4S}$ leaks $(share, sid)$ to $Sim$ which sends $s_e^{sim}$ to $B$.
            \item When $A$ is corrupted, $Sim$ verifies $dig_e^{sim} = \mathcal{H}(sid, s)$ and stores $s_e^{sim} \leftarrow s$ and sends $(share, sid)$ as honest $A$ to $\mathcal{F}_{r4S}$.
            \item For both cases, set $shared^{r4S} \leftarrow 1$.
        \end{itemize}
        \item $P\text{-}redeeming$: In the hybrid world any party $P \in \{A, B, M_1, M_2\}$ sends $(redeem$, $sid$, $tx_A^{lock}$, $s_m$, $s_e$, $s_{r1}$, $s_{br}$, $s_{r2}$, $path)$ or $(redeem$, $sid$, $tx_B^{lock}$, $s_m$, $s_{r1}$, $s_{br}$, $s_{r2}$, $path)$ such that $path \in \mathcal{Q}$ to $B$ irrespective of $P$ being honest or corrupted. In the ideal world:
        \begin{itemize}
            \item If $P$ is honest, $Sim$ internally simulates sending $(redeem$, $sid$, $tx_A^{lock}$, $s_m$, $s_e$, $s_{r1}$, $s_{br}$, $s_{r2}$, $path)$ or $(redeem$, $sid$, $tx_B^{lock}$, $s_m$, $s_{r1}$, $s_{br}$, $s_{r2}$, $path)$ to $\mathcal{G}_{mbp}$.
            \item If $P=B$, $s_m = s_m^{sim}, s_{r2} = s_{r2}^{sim}$ as $Sim$ samples these preimages for honest $B$. $s_e = s_e^{sim}$ if it was sent earlier. All the other preimages are set to $\perp$.
            \item If $P=A$, $s_e = s_e^{sim}, s_{r1} = s_{r1}^{sim}, s_{br} = s_{br}^{sim}$ as $Sim$ samples these preimages for honest $A$. The remaining preimages are set to $\perp$.
            \item If $P=M_1$ or $M_2$ all preimages are set to $\perp$.
            \item When $P$ is corrupted, $Sim$ internally simulates the message and updates $pub_m^{mbp}$, $pub_e^{mbp}$, $pub_{r1}^{mbp}$, $pub_{br}^{mbp}$, $pub_{r2}^{mbp}$ and derives $res^{mbp}$ to send back to the Party $P$. It then checks if preimages match the stored digests, that is, $dig_m^{sim} = \mathcal{H}(sid,s_m)$, $dig_{e}^{sim} = \mathcal{H}(sid,s_e)$, $dig_{r1}^{sim} = \mathcal{H}(sid,s_{r1})$, $dig_{br}^{sim} = \mathcal{H}(sid,s_{br})$ and $dig_{r2}^{sim} = \mathcal{H}(sid,s_{r2})$, for all the preimages whose digests match $Sim$ uses the influence port in the ideal functionality to update the corresponding preimages. Thus, the values of preimages remain same in both $\mathcal{F}_{r4S}$ and $\mathcal{G}_{mbp}$
        \end{itemize}
        
    \end{itemize}
    \paragraph{Output indistinguishability}
    Since the values of preimages are the same in both $\mathcal{F}_{r4S}$ and $\mathcal{G}_mbp$, the output depends on the $rPredicate$ function which is the same for both. Thus, the values of both results are the same, that is, $res_{mbp} = res_{r4S}$.
    Thus, the $Sim$ makes $\mathcal{F}_{r4S}$ computationally indistinguishable from $\Pi_{r4S}$ by simulating $\mathcal{H}$ and $\mathcal{G}_mbp$. Thus, the construction of such a $Sim$ proves that $\Pi_{r4S}$ UC-realizes $\mathcal{F}_{r4S}$.
\end{proof}

\section{Bitcoin Scripts Implementation}
\label{implementation}
We have written two locking scripts in the Bitcoin scripting language, representing the conditions for locking the assets in the 4-Swap Protocol as mentioned in Figure~\ref{4s_utxo} and Figure~\ref{4s_utxoB}. The scripts leverage the hashlock and timelock functionality of Bitcoin and are presented in Figure~\ref{fig:lockingscriptsA} and Figure~\ref{fig:lockingscriptsB} along with the different redeem paths for different scenarios.

\begin{figure}
    \centering
    \includegraphics[width=0.7\textwidth]{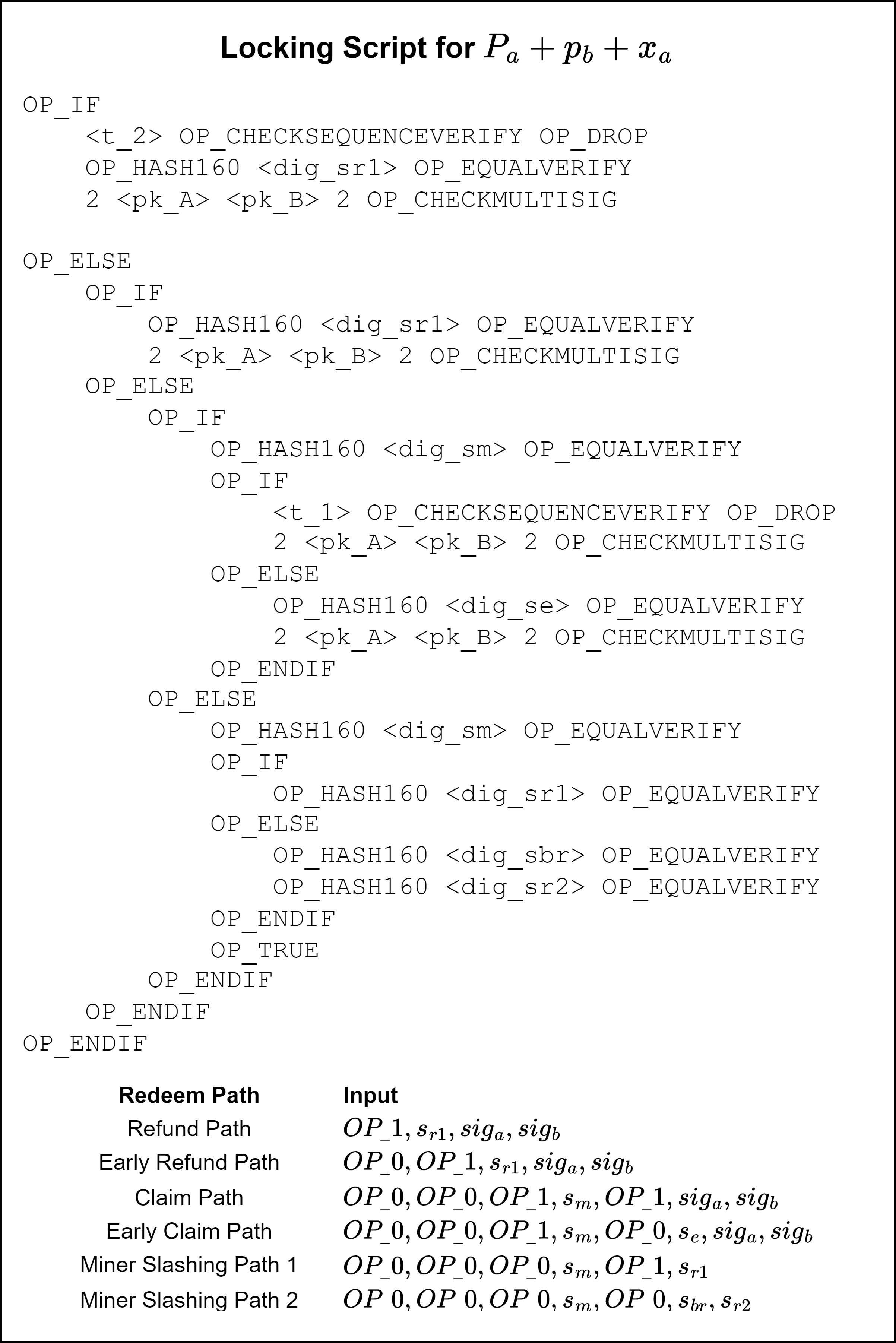}
    \caption{Locking script and redeem paths for Chain-A}
    \label{fig:lockingscriptsA}
\end{figure}

\begin{figure}
    \centering
    \includegraphics[width=0.7\textwidth]{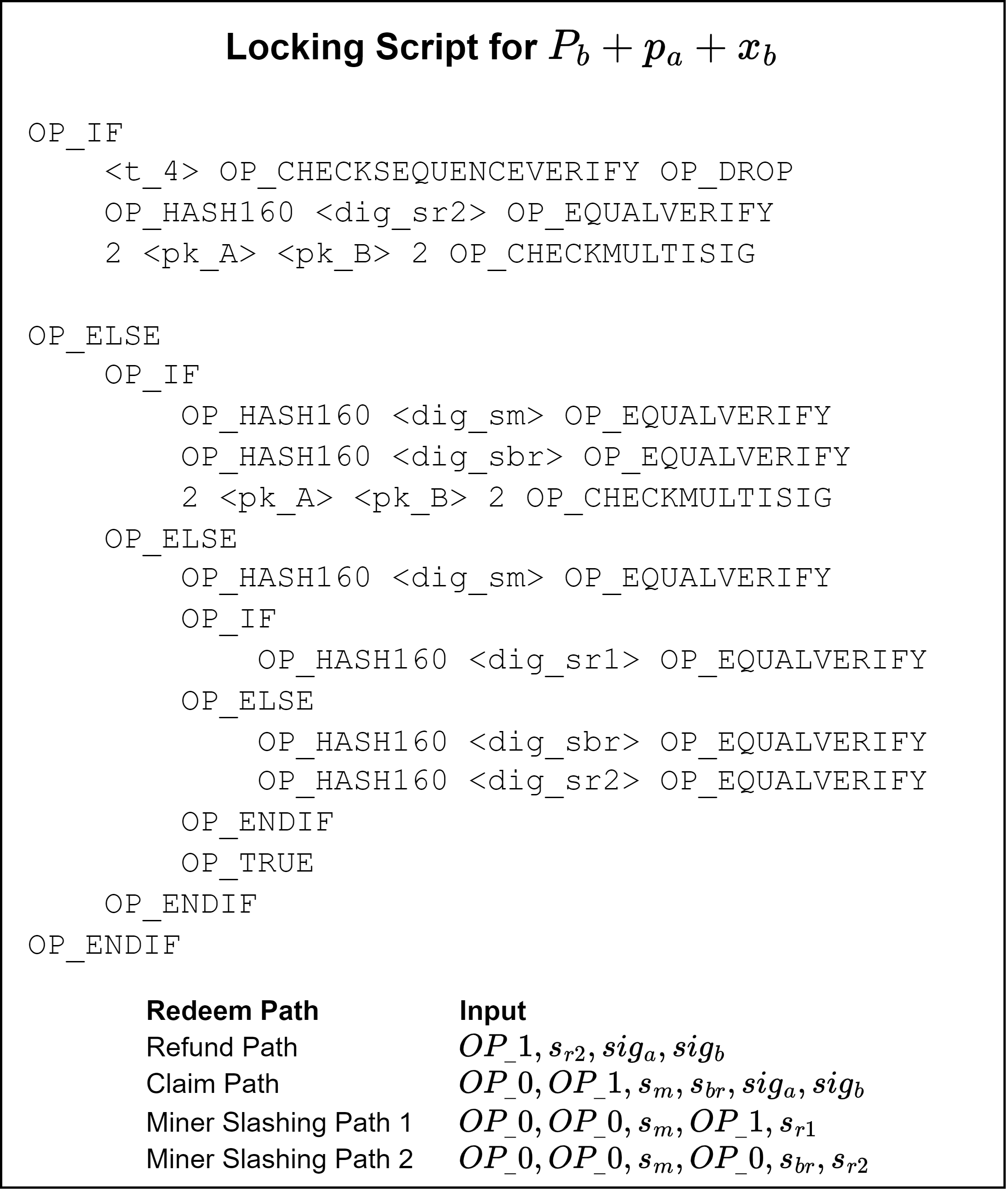}
    \caption{Locking script and redeem paths for Chain-B}
    \label{fig:lockingscriptsB}
\end{figure}

\end{document}